\documentclass[reqno]{amsart}
\usepackage{amssymb}
\usepackage{enumerate}
\usepackage{hyperref}
\usepackage{doi,url}

\usepackage{color}

\definecolor{purple}{rgb}{.5,0,1}
\definecolor{orange}{rgb}{1,.5,0}
\definecolor{pink}{rgb}{1,0,.5}

\newcommand{\nn}{\notag}

\numberwithin{equation}{section}
\newtheorem{theorem}{Theorem}[section]
\newtheorem{proposition}[theorem]{Proposition}
\newtheorem{lemma}[theorem]{Lemma}
\newtheorem{corollary}[theorem]{Corollary}
\newtheorem{definition}[theorem]{Definition}
\newtheorem{remark}[theorem]{Remark}

\newtheorem*{theorem*}{Theorem}
\newtheorem*{corollary*}{Corollary}

\DeclareMathOperator{\dist}{dist}

\DeclareMathOperator{\diam}{diam}

\makeatletter
\DeclareRobustCommand\widecheck[1]{{\mathpalette\@widecheck{#1}}}
\def\@widecheck#1#2{%
    \setbox\z@\hbox{\m@th$#1#2$}%
    \setbox\tw@\hbox{\m@th$#1%
       \widehat{%
          \vrule\@width\z@\@height\ht\z@
          \vrule\@height\z@\@width\wd\z@}$}%
    \dp\tw@-\ht\z@
    \@tempdima\ht\z@ \advance\@tempdima2\ht\tw@ \divide\@tempdima\thr@@
    \setbox\tw@\hbox{%
       \raise\@tempdima\hbox{\scalebox{1}[-1]{\lower\@tempdima\box
\tw@}}}%
    {\ooalign{\box\tw@ \cr \box\z@}}}
\makeatother

\renewcommand\H{\mathcal{H}}

\newcommand\R{\mathbb R}
\newcommand\N{\mathbb N}
\newcommand\C{\mathbb C}
\newcommand\Z{\mathbb Z}
\newcommand\G{\mathbb{G}}

\newcommand\e{\mathrm{e}}
\newcommand{\la}{\langle}
\newcommand{\ra}{\rangle}
\renewcommand\P{\mathbb P}
\newcommand\E{\mathbb E}
\newcommand\cE{\mathcal{E}}

\newcommand\cL{\mathcal{L}}
\newcommand\cB{\mathcal{B}}

\newcommand\cG{\mathcal{G}}
\newcommand\cV{\mathcal{V}}

\newcommand\cA{\mathcal{A}}

\newcommand\cF{\mathcal{F}}
\newcommand\cS{\mathcal{S}}
\newcommand\cY{\mathcal{Y}}

\newcommand\eps{\varepsilon}
\newcommand\vphi{\varphi}

\newcommand\vrho{\varrho}

\newcommand{\vs}{\varsigma}

\newcommand{\pr}{\prime}

\newcommand\wtilde{\widetilde}

\newcommand\ttau{\widetilde{\tau}}
\newcommand\tzeta{\widetilde{\zeta}}
\newcommand{\partialin}{\partial_{\mathrm{in}}}
\newcommand{\partialex}{\partial_{\mathrm{ex}}}

\newcommand{\bom}{{\boldsymbol{{\omega}}}}

\newcommand\beq{\begin{equation}}
\newcommand\eeq{\end{equation}}

\newcommand{\abs}[1]{\left\lvert #1 \right\rvert}
\newcommand{\norm}[1]{\left\lVert #1 \right\rVert}
\newcommand{\scal}[1]{\left\langle #1 \right\rangle}
\newcommand{\set}[1]{\left\{ #1 \right\}}
\newcommand{\pa}[1]{\left( #1 \right)}

\newcommand{\fl}[1]{\left\lfloor #1 \right\rfloor}
\newcommand{\br}[1]{\left [ #1 \right]}

\newcommand\La{\Lambda}

\newcommand\Th{\Theta}
\newcommand\Ups{\Upsilon}

\newcommand{\eq}[1]{\eqref{#1}}

\newcommand{\up}[1]{^{(#1)}}

\newcommand{\qtx}[1]{\quad\text{#1}\quad}
\newcommand{\mqtx}[1]{\; \ \text{#1}\; \  }
\newcommand{\sqtx}[1]{\;\text{#1}\;}

\newcommand\Chi{\raisebox{.2ex}{$\chi$}}

\begin{document}

\title[Eigensystem multiscale analysis via the Wegner estimate]{Eigensystem multiscale analysis for the Anderson model  via the Wegner estimate}

\author{Alexander Elgart}
\address[A. Elgart]{Department of Mathematics; Virginia Tech; Blacksburg, VA, 24061, USA}
 \email{aelgart@vt.edu}

\author{Abel Klein}
\address[A. Klein]{University of California, Irvine;
Department of Mathematics;
Irvine, CA 92697-3875,  USA}
 \email{aklein@uci.edu}

\thanks{A.E. was  supported in part by the NSF under grant DMS-1907435 and by  the Simons Foundation under grant 443529.}
\thanks{A.K. was  supported in part by the NSF under grant  DMS-1301641.}


\begin{abstract}
We present a new approach to the  eigensystem multiscale analysis (EMSA) for random Schr\"odinger operators that relies on the Wegner estimate.  The EMSA  treats all energies  of the finite volume operator in an energy interval at the same time,  simultaneously establishing localization of  all eigenfunctions with eigenvalues in the energy interval  with high probability. 
 It implies all the usual manifestations of localization
(pure point spectrum with exponentially decaying eigenfunctions, dynamical localization).
The new method removes the restrictive level spacing hypothesis used in the previous versions of the EMSA.   The method is presented in the context of the
 Anderson model, allowing    for  single site probability distributions that are  H\"older continuous of order  $\alpha \in (0,1]$.
  \end{abstract}
  
  \maketitle
%

\setcounter{tocdepth}{1}
\tableofcontents

\section*{Introduction}

In \cite{EK,EK2} we developed 
an eigensystem multiscale analysis (EMSA)  for proving  localization (pure point spectrum with exponentially decaying eigenfunctions, dynamical localization)    for  random Schr\"odinger operators.
The EMSA  treats all energies  of the finite volume operator in an energy interval at the same time,  simultaneously establishing localization of  all eigenfunctions with eigenvalues in the energy interval  with high probability. 
The analysis in \cite{EK,EK2} (and  its bootstrap  enhancement in \cite{KlT})  relies on a probability estimate for level spacing.  For the Anderson model with a   H\"older continuous single site probability distribution  of order  $\alpha \in (\frac 12,1]$  such an estimate is provided by   \cite[Lemma~2]{KlM}, where it is  derived from Minami's estimate \cite{M}.
(This is the  level spacing probability estimate used in \cite{EK,EK2,KlT}.) A weaker level spacing estimate is proven for the continuous Anderson model in \cite[Theorem~2.2]{DE};  it requires a covering condition for the random potential, it holds only in a certain interval at the bottom of the spectrum,  it requires the single site probability distribution to be absolutely continuous with a density that is  uniformly Lipschitz continuous and bounded below on its support, and provides weak probability estimates.  The fact that   level spacing probability estimates are not widely known, and where known require extra hypotheses,   imposes a strong limitation on the applicability of the EMSA.

The well known methods   previously developed  for proving localization for random  Schr\"o\-dinger operators are  the multiscale analysis  (MSA) (see \cite{FS,FMSS,Dr,DK,Sp,CH,FK1,FK,GKboot,Kle,BK,GKber}) and the fractional moment method (FMM) (see \cite{AM,A,ASFH,AENSS,AW}).   As opposed to the EMSA, these methods are  based on the study of   finite volume  Green's functions, and the analysis is performed  either at a fixed energy in a single box, or for all energies in an interval at once but with two boxes with an `either or' statement for each energy.   Green's functions-based  methods  do not rely on  level spacing. 
Rather, they use either explicitly   (MSA) or implicitly (FMM) a  more widely available bound,  the Wegner estimate (e.g., \cite{Weg,CH,CHK,Ki,CGK2,Kl2}).  This estimate is proven for a large family of both lattice and continuum random Schr\"odinger operators, making it possible to establish localization in these contexts.

Unfortunately, the Green's function quickly becomes an inadequate tool in the study of many-body localization,
 rending the traditional approaches to localization ineffective. The EMSA approach to localization shows more flexibility in this regard: In   a forthcoming paper, \cite{EK3}, we use the EMSA to establish  many-body localization results in the context of  random XXZ spin quantum  chains. However, as we already mentioned, the previously available version of the method uses  the  level spacing hypothesis, which (although expected) has never been proven for 
  many-body  systems so far. The main innovation of the present work is the removal of this restrictive condition, replacing it by  an argument based
   on the Wegner estimate. More precisely, the new approach uses Wegner estimates between boxes, as in \cite{FMSS,DK,GKboot,Kle}.   To illustrate the method we consider here its application to a single particle lattice Anderson model.   In this  context   it applies    when the  single-site probability distribution  is H\"older continuous of order  $\alpha \in (0,1]$,  in contrast to  the  EMSA with level spacing  of  \cite{EK,EK2} that requires  $\alpha \in (\frac 12,1]$.  Moreover, this 
 version of EMSA is expected to admit extensions  to random Schr\"odinger operators where a suitable Wegner estimate is  available, such as the  continuum Anderson model. 
\section{Definitions and results}  

A discrete  Schr\"odinger operator is   an operator  of the form 
 $H=-\Delta +V$ on $\ell^2(\Z^d)$, where $\Delta$ is the  (centered) discrete  Laplacian:  
 \begin{equation}\label{defDelta}
  (\Delta \varphi)(x):=  \sum_{\substack{y\in\Z^d\\ |y-x|=1}} \varphi(y)  \qtx{for} \varphi\in\ell^2(\Z^d),
\end{equation} 
 and  $V$ is a  bounded potential.

  \begin{definition}\label{defineAnd} The Anderson model is the
 random discrete  Schr\"odinger  operator
\beq \label{defAnd}
H_{\bom} :=  -\Delta + V_{\bom} \quad \text{on} \quad  \ell^2(\Z^d), 
\eeq 
where $V_{\bom}$ is a random potential:      $V_{\bom}(x)= \omega_x$ for  $ x \in \Z^d$, where
$\bom=\{ \omega_x \}_{x\in
\Z^d}$ is a family of independent 
identically distributed random
variables,  whose  common probability 
distribution $\mu$ has bounded support and is assumed to be H\"older continuous of order  $\alpha \in (0,1]$: 
 \beq\label{Holdercont}
S_\mu(t) \le K t^\alpha \qtx{for all} t \in [0,1],
\eeq
where $K$ is  a constant and  $S_\mu(t):= \sup_{a\in \R} \mu \set{[a, a+t]} $ is the concentration function of the measure $\mu$.
\end{definition}

 To formulate our main result we need to introduce  some additional notation. Given  $\Theta\subset \Z^d$, we let $H_\Theta$ be the restriction of $ \Chi_\Theta H\Chi_\Theta$ to $\ell^2(\Theta)$. We write $\norm{\vphi}=\norm{\vphi}_{ \ell^2(\Theta)}$ for $\vphi \in  \ell^2(\Theta)$. 
 We call $(\vphi,\lambda)$ an eigenpair for $H_\Th$ if $\vphi\in \ell^2(\Theta)$ with $\norm{\vphi}=1$,  
 $\lambda \in \R$, and $H_\Theta \vphi=\lambda \vphi$.  (In other words,    $\lambda$ is an eigenvalue for $H_\Th$ and $\vphi$ is a corresponding normalized eigenfunction.)  A collection $\set{(\vphi_j,\lambda_j)}_{j\in J}$ of eigenpairs for $H_\Th$ will be called an {\it eigensystem} for $H_\Th$ if   $\set{\vphi_j}_{j\in J}$ is an orthonormal basis for $\ell^2(\Th)$.
 If  $\Theta\subset \Z^d$ is finite,  we let
$\widetilde\sigma(H_{\Theta})$ denote the eigenvalues of $H_{\Theta}$ repeated according to multiplicity (and thought of as different points in $\widetilde\sigma(H_{\Theta})$), so an eigensystem for $H_{\Theta}$ can be rewritten as $\set{(\vphi_\lambda,\lambda)}_{\lambda\in \wtilde\sigma(H_{\Theta})}$, i.e., it can be  labeled by $\widetilde\sigma(H_{\Theta})$.

If $x=(x_1,x_2,\ldots, x_d)\in \R^d$, we set $\abs{x}=\abs{x}_2= \pa{\sum_{j=1}^dx_j^2}^{\frac 12}$   and    $\norm{x}=\abs{x}_\infty= \max_{j=1,2,\ldots,d} \abs{x_j}$. 
We consider boxes in $\Z^d$ centered at points of $\R^d$. The box in $\Z^d$  of side $L>0$  centered at $x\in \R^{d}$ is given by
\begin{align} \label{defbox}
\La_L(x)&=\La^\R_L(x)\cap \Z^d, \qtx{where}
\La^\R_L(x)= \set{y \in \R^d;\  \norm{y-x} \le  \tfrac{L}{2}}.
\end{align}
By a box  $\La_L$ we will mean a box $\La_L(x)$ for some $x\in \R^d$.  We have 
\beq
 (L-2)^{d}<  \abs  {\La_L(x)}\le  (L+1)^{d}\qtx{for all} L\ge 2 \qtx{and} x\in \R^d.
\eeq

The     EMSA  is based on the study of localized eigensystems. The relevant definitions are stated in terms of exponents  $\tau,\kappa^\pr \in (0,1)$  that will be chosen later. We use the notation $ L_{\tau}=\fl{L^{\tau}}$ for $L\ge 1$.

 \begin{definition}\label{defxmloc} Let $\La_L$ be a box, $x\in \La_L$, and  $m\ge 0$.  Then   $\vphi\in \ell^2(\La_L)$  is said to be $(x,m)$-localized if  $\norm{\vphi}=1$ and\beq\label{hypdec}
\abs{\vphi(y)}\le \e^{-m\norm{y-x}}\qtx{for all} y \in \La_L \qtx{with} \norm{y-x}\ge L_\tau.
\eeq
\end{definition}

We consider energy intervals   $I(E,A)=(E-A,E+ A)$ with center $E\in \R$ and radius  $A>0$.    (When we write  $ I(E,A)$ it will be implicit that  $E\in \R$ and  $A>0$.)
Given an interval $I=I(E,A)$, we set 
\beq\label{eq:h_I}
h_{ I}(t) =h\pa{{\tfrac{{t-E}}{{A}}}} \mqtx{for} t \in \R, \mqtx{with} h(s)=\begin{cases} 
1-s^2&\text{if} \; \; s\in (-1,1)\\0 & \text{otherwise}\end{cases}.
\eeq
Note that   $h_{ I}(t)  >0 \iff t\in I$,  which implies  $h_I=\Chi_Ih_I$.

\begin{definition}  \label{defmIloc} 
Given an energy interval  $ I=I(E,A)$,   a box $\La_L$ will be called $(m,I)$-localizing for $H$ if 
\beq\label{upbm}
L^{-\kappa^\pr} \le m   \le  \tfrac 1 2 \log \pa{1 + \tfrac {A}{4d}},
\eeq
and 
 there exists an  $(m,I)$-localized eigensystem for $H_{\La_L}$, that is,   an  eigensystem   $\set{(\vphi_\nu, \nu)}_{\nu \in \wtilde\sigma(H_{\La_L})}$ for $H_{\La_L}$ such that for all $\nu \in \wtilde\sigma(H_{\La_L})$ there is $x_\nu\in \La_L$ so $\vphi_\nu$ is $(x_\nu, m h_{ I}(\nu))$-localized.
 \end{definition}

 Given  a box $\La_\ell \subset \Theta$,   a  crucial step in our analysis
 shows that if $(\psi,\lambda)
 $ is an eigenpair for $H_{\Theta}$,  with  $\lambda \in I$  not too close to the  eigenvalues of $H_{\La_\ell}$, and  the box $\La_\ell $ is $(m,I)$-localizing for $H$, then $\psi$ is exponentially small  deep inside $\La_\ell $ (see  Lemma~\ref{lemdecay2}.).  This is    proven by expanding the values of $\psi$ inside $\La_\ell $ in terms of an $(m,I)$-localizing eigensystem  for $H_{\La_\ell}$. The problem is  we only know decay for the eigenfunctions with eigenvalues in $I$; we have no information whatsoever concerning  eigenfunctions with eigenvalues that lie outside the interval $I$.  As  in \cite{EK2},   the decay of the term containing the latter eigenfunctions  comes from the distance from the eigenvalue $\lambda$ to  the complement of the interval $I$,  and consequently   the decay rate for the  localization of an eigenfunction goes to zero  as the corresponding eigenvalue approaches the edges of the interval $I$. The introduction of the modulating function $h_I$ in the decay rate 
 models this phenomenon. 

 The control  of the term containing  eigenfunctions corresponding to eigenvalues that lie outside the interval $I$ is given by  \cite[Lemma~3.2(ii)]{EK2}, which requires the upper bound in \eq{upbm}.
 The lower bound in \eq{upbm} is a requirement for the multiscale analysis, as in  \cite{FK,GKboot,Kle,GKber}.

Our main result pertaining to  the eigensystem multiscale analysis in an energy interval is given in  the following theorem.   To state the theorem, given  exponents $ 0<\xi<\zeta<1$,  we choose
the exponents $\tau,\kappa^\pr\in (0,1)$ that appear in Definitions~\ref{defxmloc} and \ref{defmIloc},
as well as  exponents $\beta,\kappa,\varrho\in (0,1)$   and  $\gamma >1$,
satisfying the relations described in Appendix~\ref{apexp}.   In what follows, once  the exponents
$ 0<\xi<\zeta<1$  are fixed,  we always assume we choose and fix  the other exponents as in Appendix~\ref{apexp}.

 \begin{theorem}\label{thmMSA} Let $H_\bom$ be an Anderson model.   Given  $ 0<\xi<\zeta<1$,
   there exists a 
 a finite scale   $\cL= \cL(d,\xi,\zeta) $ and a constant  $C_{d}=C_{d,\xi,\zeta} >0$  with the following property:  Suppose for some scale 
$L_0 \ge \cL$ and    interval  $I_0=I(E,A_0)$  we have
  \begin{align}\label{initialconinduc9932}
\inf_{x\in \R^d} \P\set{\La_{L_0} (x) \sqtx{is}  (m_0,I_0) \text{-localizing for} \; H_{\bom}} \ge 1 -  \e^{-L_0^\zeta}.
\end{align}
 Then for all   $ L\ge L_0^\gamma $ we have 
\begin{align} \label{MSALnok2}
\inf_{x\in \R^d} \P\set{\La_{L} (x) \sqtx{is} ( m_\infty , I_\infty^L)  \text{-localizing for} \; H_{\bom}} \ge 1 -  \e^{-L^\xi}  ,
\end{align}
where
\begin{align}\label{Aminfty} 
  I_\infty^L& =  I_\infty^L(L_0)= I(E,{A_\infty}(1-L^{-\kappa})^{-1}), \\ \nn
  A_\infty &=A_\infty (L_0)= A_0 \prod_{k=0}^\infty \pa{1- L_0^{-\kappa\gamma^k}}, \\ \nn
   L_0^{-\gamma \kappa^\pr} \le \  m_\infty&=m_\infty (L_0)= m_0 \prod_{k=0}^\infty \pa{1- C_{d} L_0^{-\vrho\gamma^k}} <   \tfrac 1 2 \log \pa{1 + \tfrac {A_\infty}{4d}}.   \end{align}
In particular,  
$\lim_{L_0\to \infty} A_\infty (L_0) = A_0$ and $\lim_{L_0\to \infty} m_\infty (L_0) = m_0 $.
 \end{theorem}

We now state a corollary of Theorem~\ref{thmMSA} that encapsulates the usual forms of Anderson localization (pure point spectrum with exponentially decaying eigenfunctions, dynamical localization, etc.) on the interval $I_\infty=I(E,A_\infty)$, as in  \cite{GKsudec,GKber,EK}.  
We fix $\nu > \frac d 2$, and  given $a \in \Z^{d}$ we define
 $T_{a} $ as  the operator on   $\ell^2(\Z^d)$ given by multiplication by the function
$T_{a}(x):= \la x-a\ra^{\nu}$, where  $ \la x\ra=  \sqrt{1 + \norm{x}^2}$. Since   $\langle a +b \rangle \le \sqrt{2}\langle a \rangle
\langle b\rangle$, we have $
\| T_{a} T_{b}^{-1} \| \le 2^{\frac  {\nu}  2} \la  a -b \ra^{\nu}$.
A function
$\psi\colon \Z^d \to \C$ is a $\nu$-generalized eigenfunction for  the discrete  Schr\"odinger operator $H$ if $\psi$ is a generalized eigenfunction 
and $\norm{T_0^{-1} \psi}<\infty$.  ($\norm{T_0^{-1} \psi}<\infty$ if and only if  $\norm{T_a^{-1} \psi}<\infty$ for all $a\in \Z^d$.)
We let 
$\cV({\lambda})$ denote the collection of  $\nu$-generalized eigenfunctions for $H$ with generalized eigenvalue ${\lambda} \in \R$.
Given   ${\lambda} \in \R$ and $a,b \in \Z^{d}$, we set
 \begin{align} \label{defGWx}
W_{\lambda}\up{a}({b}):=\begin{cases} 
\sup_{\psi \in\cV({\lambda}) }
\ \frac {\abs{\psi(b)}}
{\|T_{a}^{-1}\psi \|}&
 \text{if $\cV({\lambda})\not=\emptyset$}\\0 & \text{otherwise}\end{cases}.
\end{align}
For  all $a,b,c \in \Z^d$ we have 
\begin{equation}\label{boundGW}
W\up{a}_{\lambda}({a})\le 1,\; W\up{a}_{\lambda}({b})\le\la b-a\ra^\nu,  \sqtx{and} W\up{a}_{\lambda}({c})\le 2^{\frac \nu 2} \la b-a\ra^\nu 
W\up{b}_{\lambda}({c}).
\end{equation}

 \begin{corollary}\label{thmloc}   Suppose the conclusions of  Theorem~\ref{thmMSA} hold  for an  Anderson model $H_{\bom}$, and  let $I=I_\infty$, $m=m_\infty$.   There is a finite scale   $\cL=\cL_{d,\nu} $ such that, given  $\cL\le L \in 2\N$ and  $a\in \Z^d$,    there exists an event  $\cY_{L,a}$ with the following properties:
  
  \begin{enumerate}
\item $\cY_{L,a}$  depends only on the random variables $\set{\omega_{x}}_{x \in \Lambda_{5L}(a)}$  and  
\beq\label{cUdesiredint}
  \P\set{\cY_{L,a} }\ge  1 - C \e^{-L^\xi}.
  \eeq

\item Given  $\bom \in \cY_{L,a}$,   for all  $\lambda \in I$ we have that  
\beq \label{locimpl}
\max_{b\in \La_{\frac L 3}(a)} W\up{a}_{\bom,\lambda}(b)>\e^{-\frac 1 4 m h_{ I^L} (\lambda) L} \  \Longrightarrow \ \max_{y\in A_L(a)} W\up{a}_{\bom,\lambda}(y)\le \e^{-\frac 7 {132}m  h_{ I^L}  (\lambda) \norm{y-a}},
\eeq
where
 \beq\label{Aell}
 A_L(a):=  \set{y\in \Z^d; \ \tfrac 8 7 L \le   \norm{y-a}\le \tfrac {33}{14} L}.
  \eeq
In particular, for all $\bom \in \cY_{L,a}$ and   $\lambda \in I$ we have
\beq  \label{WW}
W\up{a}_{\bom,\lambda}(a)W\up{a}_{\bom,\lambda}(y)\le 
 \e^{- \frac 7 {132}  m h_{ I^L}  (\lambda)  \norm{y-a}}\mqtx{for all} y\in A_L(a).
\eeq

   \end{enumerate}
   \end{corollary}

   Although Corollary~\ref{thmloc} looks exactly like \cite[Theorem~1.7]{EK2},  Theorem~\ref{thmMSA}
is not the same as \cite[Theorem~1.6]{EK2} (the definitions of a localizing box are different,  the conclusion \eq{MSALnok2} is stated differently from \cite[Equation~(1.20)]{EK2}). For this reason  the derivation of 
   Corollary~\ref{thmloc} from  Theorem~\ref{thmMSA} has some differences from the derivation of 
    \cite[Theorem~1.7]{EK2} from   \cite[Theorem~1.6]{EK2}, so it is included in this paper.

The usual forms of localization can be derived from Corollary~\ref{thmloc}  and  are stated in the  following corollary.

\begin{corollary} \label{corloc}  Suppose the conclusions of  Theorem~\ref{thmMSA} hold  for an  Anderson model $H_{\bom}$, and  let $I=I_\infty$, $m=m_\infty$. Then the following holds with probability one:
  \begin{enumerate}
 \item  {$H_{\bom}$}  has  pure point spectrum in the interval $I$.

\item   If   {$\psi_\lambda$}  is an {eigenfunction} of $H_{\bom}$
with eigenvalue  {$\lambda\in I$}, then $\psi_\lambda$ is exponentially localized  with rate of decay $ \frac 7 {132} m  h_{I}(\lambda)$,   more precisely,
\begin{equation}\label{expdecay222}
\abs{\psi_\lambda(x)} \le C_{\bom,\lambda}\norm{T_0^{-1} \psi}\, e^{-  \frac 7 {132}m  h_{I}(\lambda)\norm{x}} \qquad \text{for all}\quad  x \in \R^{d}.
\end{equation}

 \item If  $\lambda \in I$, then for all   $x,y \in \Z^d$ we have 
 \begin{align}\label{eqWW}
W\up{x}_{\bom,\lambda}(x)W\up{x}_{\bom,\lambda}(y)
\le   C_{m,\bom,\nu} \pa{h_{I} (\lambda)}^{-\nu}\! \e^{(\frac 4 {33} +\nu)m  h_{I} (\lambda)   (2d\log \scal{x})^{\frac 1 \xi}}\!  \e^{- \frac 7 {132} m  h_{I} (\lambda) \norm{y-x}} .
\end{align}

\item  If  $\lambda \in I$, then for  $ \psi\in \Chi_{\set{\lambda}}(H_{\bom})$ and all $x,y \in \Z^d$ we have
\begin{align}\label{eqWW2}
 &\abs{\psi(x)}\abs{\psi(y)}\\ \notag
& \
\le  C_{m,\bom,\nu} \pa{h_{I} (\lambda)}^{-\nu}\, \norm{T_x^{-1} \psi}^2\e^{(\frac 4 {33} +\nu)m  h_{I} (\lambda)   (2d\log \scal{x})^{\frac 1 \xi}}  \e^{- \frac 7 {132} m  h_{I} (\lambda) \norm{y-x}}  \\ \notag
&\  \le 2^\nu C_{m,\bom,\nu} \pa{h_{I} (\lambda)}^{-\nu}\, \norm{T_0^{-1} \psi}^2\la x\ra^{2\nu}\e^{(\frac 4 {33} +\nu)m  h_{I} (\lambda)   (2d\log \scal{x})^{\frac 1 \xi}}  \e^{- \frac 7 {132} m  h_{I} (\lambda) \norm{y-x}} .
 \end{align}

\item If  $\lambda \in I$, then there exists $x_\lambda=x_{\bom,\lambda} \in \Z^d$, such that for $ \psi\in \Chi_{\set{\lambda}}(H_{\bom})$ and  all $x \in \Z^d$ we have
\begin{align}\nn
&\abs{\psi(x)} 
 \le C_{m,\bom,\nu} \pa{h_{I} (\lambda)}^{-\nu}\!\norm{T_{x_\lambda}^{-1} \psi} \e^{(\frac 4 {33} +\nu)m  h_{I} (\lambda)   (2d\log \scal{x_\lambda})^{\frac 1 \xi}}  \e^{- \frac 7 {132} m  h_{I} (\lambda) \norm{x-x_\lambda}} \\ 
  &   \le  
 2^{\frac \nu 2} C_{m,\bom,\nu} \pa{h_{I} (\lambda)}^{-\nu}\!\norm{T_0^{-1} \psi}\!\la x_\lambda\ra^{\nu}\e^{(\frac 4 {33} +\nu)m  h_{I} (\lambda)   (2d\log \scal{x_\lambda})^{\frac 1 \xi}}\!  \e^{- \frac 7 {132} m  h_{I} (\lambda) \norm{x-x_\lambda}} .
 \end{align}

\end{enumerate}
\end{corollary}

 In Corollary~\ref{corloc}, (i) and (ii)  are statements of Anderson localization, (iii) and (iv) are statements of dynamical localization ((iv) is called   SUDEC (summable uniform decay of eigenfunction correlations) in \cite{GKsudec}), and (v) is SULE (semi-uniformly localized eigenfunctions; see
\cite{DRJLS0,DRJLS,GKsudec}).
Statements of localization in expectation can also be derived, as in \cite{GKsudec,GKber}.

The proof of Corollary~\ref{corloc} from  Corollary~\ref{thmloc} is the same as the proof of \cite[Corollary~1.8]{EK2}  from \cite[Theorem~1.7]{EK2}, with some obvious modifications,  so we refer to \cite{EK2}.

Theorem~\ref{thmMSA}  also implies localization at the bottom of the spectrum as in  \cite[Section 2]{EK2}.

The conclusions of Theorem~\ref{thmMSA}  are equivalent to the conclusions of the energy interval multiscale analysis \cite{FMSS,DK,GKboot,Kle}; this can be seen  proceeding  as in \cite[Section~6]{EK2}. Finally, we stress that the theorem holds for Anderson models whose single-site probability distributions  satisfy \eqref{Holdercont}.

In   the remainder of this paper  we fix  $ 0<\xi<\zeta<1$ and the corresponding 
  exponents $\tau, \beta,\kappa,\kappa^\pr, \varrho\in (0,1)$   and  $\gamma >1$,
as in  Appendix~\ref{apexp}.   The deterministic lemmas for the EMSA are introduced in Section~\ref{secprep}.  The probability estimates based on Wegner estimates are presented in Section~\ref {secprobest}.
Theorem~\ref{thmMSA} is proven in Section~\ref{secEMSA}.    The proof of Corollary~\ref{thmloc} is 
given in Section~\ref{seclocproof}.

\section{Lemmas for  the eigensystem multiscale analysis}\label{secprep}

In this section we introduce notation and deterministic   lemmas that will  play an important role in the eigensystem multiscale analysis.  By $H$ we always denote  a  discrete  Schr\"odinger operator
  $H=- \Delta +V$ on $\ell^2(\Z^d)$. We also fix an interval $I=I(E,A)$.

\subsection{Preliminaries}
 
  Let  $\Phi \subset \Theta\subset \Z^d$.     We define  the boundary, exterior boundary, and interior boundary of $\Phi$ relative to $\Theta$, respectively,  by
 \begin{align}\label{defbdry}
  \boldsymbol{ \partial}^{ \Theta} \Phi &=\set{(u,v) \in \Phi\times \pa{\Theta\setminus \Phi}; \  \abs{u-v}=1},
   \\
 \partial_{\mathrm{ex}}^{ \Theta} \Phi &=\set{v \in\pa{\Theta\setminus \Phi}; \ (u,v) \ \in  \boldsymbol{ \partial}^{ \Theta} \Phi\qtx{for some}u \in \Phi},\notag
    \\
 \partial^{ \Theta}_{\mathrm{in}}\Phi &=\set{u \in {\Phi}; \ (u,v) \ \in  \boldsymbol{ \partial}^{ \Theta} \Phi\qtx{for some}v \in \Theta\setminus \Phi}.\notag
   \end{align}
If $t\ge 1$, we let
\begin{align}\label{defLatTh} 
 \Phi^{\Th,t}& = \set{y\in \Phi;   \; \dist \pa{y,{ \Th}\setminus \Phi }> \fl{t}} \qtx{and}
 {\partial}_{\mathrm{in}}^{\Th,t} \Phi   = \Phi \setminus  \Phi^{\Th,t}.
 \end{align}
  We use the notation  
\beq \label{Ry}
R_{\Theta}(y) = \dist \pa{y, \partial^{ \Theta}_{\mathrm{in}}\Phi} \qtx{for} y \in \Phi.
\eeq

For a box   $\La_L\subset \Th \subset \Z^d$ we  write   $\La_L^{\Th,t}(x)= \pa{\La_L(x)}^{\Th,t}$. 
For $ L\ge 2$  
 we have  
  \beq\label{bdryest}
\abs{\partial_{\mathrm{in}}^{ \Th} \La_L }\le\abs{\partial_{\mathrm{ex}}^{ \Th} \La_L } =\abs{ \boldsymbol{ \partial}^{ \Th} \La_L }\le s_{d} L^{d-1}, \qtx{where} s_d= 2^{d} d.
\eeq
For $v \in \Th$ we let   $\hat{v} \in   \partial_{\mathrm{in}}^\Th { \La_L} $ be the unique $u \in   \partial_{\mathrm{in}}^\Th { \La_L} $
 such that $(u,v)\in \boldsymbol{\partial}^\Th { \La_L} $ if $v\in \partial_{\mathrm{ex}}^\Th { \La_L} $,
 and set $\hat v=0$ otherwise.

If $\Phi \subset \Theta\subset \Z^d$, we consider $\ell^2(\Phi)\subset \ell^2(\Theta)$ by extending functions on $\Phi$ to functions on $\Theta$ that are identically $0$ on $\Theta\setminus \Phi$.  
We have
\begin{gather}\label{Hdecomp1}
H_{ \Theta}= H_{ \Phi}\oplus H_{ \Theta\setminus  \Phi} + \Gamma_{ \boldsymbol{ \partial}^{ \Theta}  \Phi} 
\qtx{on} \ell^2( \Theta)=\ell^2( \Phi)\oplus \ell^2( \Theta\setminus \Phi), \\ \nn
\text{where}\quad \Gamma_{   \boldsymbol{ \partial}^{ \Theta}  \Phi}(u,v)=
\begin{cases}
-1 & \text{if either}\  (u,v) \sqtx{or}(v,u) \in   \boldsymbol{ \partial}^{ \Theta}  \Phi\\
\ \ 0 & \text{otherwise}
\end{cases}.
\end{gather}
Given  $J \subset \R$, we set  $\sigma_J(H_\Theta)= \sigma(H_\Theta) \cap J$ and  $\wtilde\sigma_J(H_\Theta)= \wtilde\sigma(H_\Theta) \cap J$.

A function  $\psi\colon{\Theta} \to \C$ is called a generalized eigenfunction for $H_{\Theta}$ with generalized eigenvalue $\lambda \in \R$, and $(\psi,\lambda)$ is called a generalized eigenpair for $H_{\Theta}$, if $\psi$ is not identically
 zero and 
 \beq \label{pointeig} 
 \scal{(H_\Th -\lambda)\vphi, \psi}=0 \qtx{for all} \vphi \in \ell^2(\Th)\quad \text{with finite support}.
 \eeq

\begin{lemma}\label{outbad} Let  $ \Theta\subset\Z^d $ and  let  $(\psi,\lambda)$ be a generalized eigenpair for $H_{\Theta}$. Let 
$\Phi \subset \Theta$ finite, $\eta>0$, and suppose
\beq\label{distPhi0}
\dist\pa{\lambda, \sigma(H_{\Phi})} \ge \eta.
\eeq
Then for all $y \in \Phi$ we have
\beq\label{distPhi3}
\abs{\psi(y) } \le 2d \eta^{-1} \abs{\partial_{\mathrm{ex}}^{ \Theta}{\Phi}}^{\frac 12}\abs{\psi(y_1) } \qtx{for some}  y_1\in \partial_{\mathrm{ex}}^{{\Theta}}  \Phi.
\eeq
 The estimate   \eq{distPhi3} also holds (trivially)  for   $y \in  \partial_{\mathrm{ex}}^{{\Theta}}  \Phi$ if $ 2d \eta^{-1}\ge 1$.
\end{lemma}

\begin{proof} Let  $\set{(\vphi_\nu,\nu)}_{\nu \in \wtilde\sigma(H_\Phi)}$ be an eigensystem for $H_\Phi$. 
If $\nu \in \wtilde\sigma(H_\Phi) $,  we have  
   $\abs{\lambda -\nu}\ge \eta$ by \eq{distPhi0}.
 Since  $\Phi$ is finite,  using  \eq{pointeig}  and \eq{Hdecomp1} we get 
\begin{align}
\scal{ \vphi_\nu,\psi}&= \pa{\lambda- \nu }^{-1}\scal{\pa{H_{\Theta}- \nu} \vphi_\nu, \psi} = \pa{\lambda- \nu }^{-1}\scal{\pa{H_{\Theta}- H_\Phi} \vphi_\nu, \psi} \\ \nn
&= \pa{\lambda- \nu }^{-1}\scal{ \vphi_\nu,  \Gamma_{ \boldsymbol{ \partial}^{ \Theta}  \Phi}\psi} .
\end{align}
It follows that for $y\in \Phi$ we have  
\begin{align}\label{sum00099}
&\psi(y)  = \scal{\delta_y,\psi} =   \scal{\delta_y, \sum_{\nu \in \wtilde\sigma(H_{\Phi})}\scal{ \vphi_\nu,\psi}{\vphi_\nu} }\\
& \ \nn =  \scal{\delta_y, \sum_{\nu \in \wtilde\sigma(H_{\Phi})}\pa{\lambda- \nu }^{-1}\scal{ \vphi_\nu,  \Gamma_{ \boldsymbol{ \partial}^{ \Theta}  \Phi}\psi}{\vphi_\nu} }\\
& \  \nn=  \scal{\delta_y, \sum_{\nu \in \wtilde \sigma(H_{\Phi})}\pa{\lambda- \nu }^{-1}\scal{ \vphi_\nu,  \Chi_\Phi\Gamma_{ \boldsymbol{ \partial}^{ \Theta}  \Phi}\psi}{\vphi_\nu} }\\
& \ \nn 
=  \scal{\delta_y,\pa{\lambda- H_{\Phi} }^{-1} \sum_{\nu \in \wtilde\sigma(H_{\Phi})}\scal{ \vphi_\nu,  \Chi_\Phi\Gamma_{ \boldsymbol{ \partial}^{ \Theta}  \Phi}\psi}{\vphi_\nu} }=  \scal{\delta_y,\pa{\lambda- H_{\Phi} }^{-1} \Chi_\Phi \Gamma_{ \boldsymbol{ \partial}^{ \Theta}  \Phi}\psi}.
\end{align}
Using  \eq{distPhi0}, we get
\begin{align}
\abs{\psi(y) }&\le \eta^{-1} \norm{\Chi_\Phi\Gamma_{ \boldsymbol{ \partial}^{ \Theta}  \Phi}\psi}= \eta^{-1} \norm{\Chi_\Phi\Gamma_{ \boldsymbol{ \partial}^{ \Theta}  \Phi}\Chi_{\partial_{\mathrm{ex}}^{ \Theta}{\Phi}}\psi}\le 2d  \eta^{-1} \norm{\Chi_{\partial_{\mathrm{ex}}^{ \Theta}{\Phi}}\psi}\\ \nn &
\le 2d  \eta^{-1}\abs{\partial_{\mathrm{ex}}^{ \Theta}{\Phi}}^{\frac 12}\abs{\psi(y_1) } \qtx{for some}  y_1\in \partial_{\mathrm{ex}}^{{\Theta}}  \Phi.
\end{align}
\end{proof}

For the  interval $I=I(E,{A})$ and   $L>1 $, we set
\begin{align}\label{defIell}
I_L= I(E,{A}(1-L^{-\kappa})) \subsetneq I=I(E,{A})  \subsetneq
I^L= (E,{A}(1-L^{-\kappa})^{-1}).
\end{align}
We write $ I_L^{L^\pr} = \pa{I_L}^{L^\pr} = \pa{I^{L^\pr}}_L$, and observe that  $I_L^L=I$.  Note that 
\beq\label{lowerbdh}
h_{ I}(t)\ge 1- (1-L^{-\kappa})^2\ge L^{-\kappa} \mqtx{for all} t\in I_L, \qtx{so} h_I \Chi_{I_L}   \ge   L^{-\kappa} \Chi_{I_L}.
\eeq

\subsection{Localizing boxes} The following lemma plays a crucial role in the multiscale analysis.
 It says that given an eigenpair  $(\psi,\lambda)
 $  for $H_{\Theta}$ and 
 a box $\La_\ell \subset \Theta$    
 with  $\lambda \in I_\ell$  not too close to the  eigenvalues of $H_{\La_\ell}$,  then $\psi$ is exponentially small  deep inside $\La_\ell $ if  the box $\La_\ell $ is $(m,I)$-localizing for $H$. 

 If $\La_\ell$ is an  $(m,I)$-localizing box,   $\set{(\vphi_\nu, \nu)}_{\nu \in \wtilde\sigma(H_{\La_\ell})}$ will denote  an  $(m,I)$-localized eigensystem for $H_{\La_\ell}$.
If $\La_\ell\subset \Theta \subset \Z^d$, $J\subset I$ and  $t>0$, 
 we set
\beq
\wtilde\sigma_{J}^{\Th,{t}}(H_{\La_\ell})= \set{\nu\in \wtilde\sigma_{J} (H_{\La_\ell}); \; x_\nu \in    \La_\ell^{\Th,{t}}} .
\eeq

Given a scale $\ell\ge 1$,  we set
$L=\ell^\gamma$. The exponent $\ttau$ is defined in \eq{ttauzeta2}.
 We  use the notation $ L_{\tau}=\fl{L^{\tau}}$ and $L_{{\ttau}}= \lfloor{L^{{\ttau}}}\rfloor$.

  \begin{lemma}\label{lemdecay2}   Let $\psi\colon{\Theta}\subset \Z^d  \to \C$ be a generalized eigenfunction for $H_{\Theta}$ with generalized eigenvalue   $\lambda \in I_\ell$.
Consider a box  $ \La_\ell \subset{\Theta}$  such that $\La_\ell$ is   $(m,I)$-localizing for $H$.   Suppose
\beq\label{distpointeig}
\dist\pa{\lambda, \sigma_{I}(H_{\La_\ell})}\ge  \tfrac 1 2\e^{-L^\beta}.\eeq  
 Then, if  $\ell$  is  sufficiently large,  
for all $y\in \La_\ell^{\Theta,\ell_{\ttau}}$ we have
\beq\label{decayest12}
\abs{\psi(y)}\le  \e^{-m_3 h_{I}\pa{\lambda} {R_{ \Theta}(y)}}  \abs{\psi(v)}\qtx{for some} v \in\partial_{\mathrm{ex}}^{{\Theta}} \Lambda_{\ell},
\eeq
where
\beq\label{m4} 
m_3= m_3(\ell)\ge  m \pa{1 - C_{d}\ell^{-\frac{1- \tau}2}}.  
\eeq
  \end{lemma}

Lemma~\ref{lemdecay2}  resembles \cite[Lemma~3.4]{EK2}, but the hypothesis \eq{distpointeig} is stronger than the corresponding hypothesis   \cite[Eq.~(3.24)]{EK2},  so the proof is slightly easier, and the conclusions are slightly stronger. The main issue in the proof is the same: the hypothesis that the box  $ \La_\ell \subset{\Theta}$   is   $(m,I)$-localizing only gives decay for eigenfunctions with eigenvalues in $I$.  To compensate, we take $\lambda \in I_\ell$,
 and use  \cite[Lemmas~3.2 and 3.3]{EK2}.

 \begin{proof}[Proof of  Lemma~\ref{lemdecay2}]
 
Given $y \in \La_\ell$ and $t>0$, it follows from  \cite[Lemma~3.2(i)]{EK2} that
\begin{align}\label{eq:2terms}
  {\psi(y)} &= \scal { \e^{-t \pa{\pa{H_{\La_\ell}-E}^2-(\lambda-E)^2 }}\delta_y,\psi}   -  \scal {F_{t,\lambda-E} (H_{\La_\ell}-E)\delta_y,\Gamma_{ \boldsymbol{ \partial}^{ \Theta}  \La_\ell} \psi} ,
  \end{align}
 where $\Gamma_{ \boldsymbol{ \partial}^{ \Theta}  \Phi}$ is defined in \eq{Hdecomp1} and
  $F_{t,\lambda}(z)$ is the entire function given by
  \beq\label{defanf}
  F_{t,\lambda} (z) =  \frac{1-\e^{-t(z^2-\lambda^2)}}{z-\lambda} \qtx{for} z \in \C\setminus \set{\lambda} \qtx{and} F_{t,\lambda}(\lambda)=2t\lambda.
  \eeq

We take $E=0$ by replacing the potential $V$ by $V-E$.
Setting  $P_{I}= \Chi_{I}\pa{H_ {\Lambda_{\ell}}}$ and $\bar P_{I}= 1-P_{I}$, 
we have
\begin{align}\label{Jdecomp}
\scal {\e^{-t \pa{H_{\La_\ell}^2-\lambda^2}}\delta_y,\psi} =   \scal {\e^{-t \pa{H_{\La_\ell}^2-\lambda^2}}P_{I}\delta_y,\psi} +  \scal {\e^{-t \pa{H_{\La_\ell}^2-\lambda^2}}\bar P_{I}\delta_y,\psi}.
\end{align}
It follows from \cite[Lemma~3.3]{EK2} that
\begin{align}\label{barJest55}
&\abs{\scal {\e^{-t \pa{H_{\La_\ell}^2-\lambda^2}}\bar P_{I}\delta_y,\psi}}\le \norm{\Chi_{\La_\ell} \psi}\norm{\e^{-t \pa{H_{\La_\ell}^2-\lambda^2}}\bar P_{I}}\le (\ell+1)^{\frac d 2} \e^{-t {A}^2 h_{I}(\lambda) }\abs{\psi(v)},
\end{align}
for some $v \in \La_\ell$. Estimating  $\abs{\psi(v)}$  by Lemma~\ref{outbad}, we get
\begin{align}\label{barJest}
&\abs{\scal {\e^{-t \pa{H_{\La_\ell}^2-\lambda^2}}\bar P_{I}\delta_y,\psi}}\le 4d \pa{s_d \ell^{{d-1}}}^{\frac 12} (\ell+1)^{\frac d 2} \e^{L^\beta}\e^{-t {A}^2 h_{I}(\lambda) }\abs{\psi(v_0)}\\
\nn &\qquad  \le  \e^{2L^\beta}\e^{-t {A}^2 h_{I}(\lambda) }\abs{\psi(v_0)}, \qtx{for some}v_0 \in\partial_{\mathrm{ex}}^{{\Theta}} \Lambda_{\ell}.
\end{align}

We now use the fact that   $ \La_\ell $  is   $(m,I)$-localizing for $H$,  so it has an $(m,I)$-localized eigensystem $\set{\vphi_\nu,\nu}_{\nu \in \wtilde\sigma{(H_{\La_\ell})}}$, and write   
 \begin{align}
\scal {\e^{-t \pa{H_{\La_\ell}^2-\lambda^2}}\ P_{I}\delta_y,\psi}&= \sum_{\nu \in\wtilde \sigma_{I}(\H_{\La_\ell})}\e^{-t (\nu ^2-\lambda^2)}{\vphi_\nu (y)} \scal{ \vphi_\nu ,\psi}.
\end{align}
If $\nu \in  \wtilde\sigma_I(H_{\Lambda_{\ell}}) $, we have  
   $\abs{\lambda -\nu}\ge \tfrac 1 2\e^{-{L^\beta}}$ by \eq{distpointeig}.
 Since  $\La_\ell$ is finite,   \eq{pointeig}  gives
\beq
\scal{ \vphi_\nu,\psi}= \pa{\lambda- \nu }^{-1}\scal{\pa{H_{\Theta}- \nu} \vphi_\nu, \psi} .
\eeq
It follows from  \cite[Eq.~(3.12) in Lemma~3.2]{EK}  that
\begin{align}\label{vphipsi399}
\abs{{\vphi_{\nu}(y)}\scal{ \vphi_{\nu},\psi}}& \le 2\e^{{L^\beta}}   \sum_{  v \in   \partial_{\mathrm{ex}}^{{\Theta}}  \Lambda_{\ell}}\abs{ {\vphi_{\nu}(y)}\vphi_{\nu} (\hat{v})}\abs{\psi(v)}.
\end{align}

We  now assume   $y\in \La_\ell^{\Theta,\ell_{\ttau}}$, so $ {R_{ \Theta}(y)}\ge \ell_{\ttau}$.  For ${\nu} \in\wtilde \sigma_{I}^{\Th,{\ell_\tau}}(H_{\La_\ell})$  and $v^\pr \in \partialin^{{\Theta}}  \Lambda_{\ell}$, we have, as in \cite[Eq. (3.41)]{EK}, 
\beq\label{mipr}
\abs{{\vphi_{\nu}(y)}\vphi_{\nu}(v^\pr)}\le  \e^{-m_1^\pr   h_I({\nu})  {R_{ \Theta}(y)}} \qtx{with}  m^\pr_1\ge   m(1-  2  \ell^{\frac {\tau -1}2}),
\eeq
so, as in \cite[Eq. (3.44)]{EK}, for ${\nu} \in  \wtilde\sigma_I^{\Th,{\ell_{\tau}}}(H_{\Lambda_{\ell}}) $ we have\begin{align}
\abs{{\vphi_{\nu}(y)}\scal{ \vphi_{\nu},\psi}}\le 2 \e^{L^\beta} s_d\ell^{d-1}  \e^{-m_1^\pr   h_I({\nu})   {R_{ \Theta}(y)}}\abs{\psi(v_1)}\le \e^{2L^\beta}   \e^{-m_1^\pr   h_I({\nu})  {R_{ \Theta}(y)}}\abs{\psi(v_1)},
\end{align}
for some $v_1 \in\partialex^{{\Theta}}  \Lambda_{\ell}$.  If ${\nu} \in \wtilde\sigma_{I} (H_{\La_\ell})$ with  $x_{\nu} \in  \partialin^{\Th,\ell_\tau} \La_\ell$, we have 
\beq\norm{x_{\nu}-y}\ge {R_{ \Theta}(y)} - \ell_{\tau} \ge   {R_{ \Theta}(y)}\pa{1- 2 \ell^{\tau- \ttau}}= {R_{ \Theta}(y)}\pa{1- 2 \ell^{\frac {\tau-1}2}},
\eeq
 so 
\begin{align}
&\abs{{\vphi_\nu(y)} \scal{ \vphi_\nu,\psi}} \le \e^{-m   h_I({\nu}) \norm{x_{\nu}-y}}\norm{\Chi_\La \psi} \\ \notag & \quad
\le \e^{-m   h_I({\nu}) {R_{ \Theta}(y)}\pa{1- 2 \ell^{\frac {\tau-1}2}}} (\ell+1)^{\frac d 2} \abs{\psi(v_2)}
\le  (\ell+1)^{\frac d 2} \e^{-m_1^\pr   h_I({\nu}) {R_{ \Theta}(y)}}  \abs{\psi(v_2)} ,
\end{align}
for some  $v_2 \in \La_\ell$, where $m_1^\pr$ is given in \eq{mipr}.
 It follows that for all ${\nu} \in \wtilde \sigma_{I}(H_{\La_\ell})$ we have
\begin{align}
\e^{-t ({\nu}^2-\lambda^2)}\abs{{\vphi_{\nu}(y)}\scal{ \vphi_{\nu},\psi}}\le  \e^{2L^\beta} 
\e^{-t ({\nu}^2-\lambda^2)}  \e^{-m_1^\pr   h_I({\nu}) {R_{ \Theta}(y)}}\abs{\psi(v)} ,
\end{align}
for some $ v\in \La_\ell\cup \in\partialex^{{\Theta}}  \Lambda_\ell$.

We now take  
\beq
t= \tfrac{m_1^\pr {R_{ \Theta}(y)}}{{A}^2}  \quad \Longrightarrow \quad \e^{-t ({\nu}^2-\lambda^2)}\e^{-m_1^\pr h_I ({\nu}){ {R_{ \Theta}(y)}}}= \e^{-m_1^\pr h_I (\lambda) {R_{ \Theta}(y)}}\mqtx{for} {\nu} \in I,
\eeq
obtaining
\begin{align}\label{Jest32}
& \abs{\scal {\e^{-\tfrac{m_1^\pr {R_{ \Theta}(y)}}{{A}^2}  \pa{H_{\La_\ell}^2-\lambda^2}} P_{I}\delta_y,\psi}}\le   (\ell+1)^d \e^{2L^\beta} \e^{-m_1^\pr h_I (\lambda){ {R_{ \Theta}(y)}}}\abs{\psi(v)}\\ \notag &\; \le  4d \pa{s_d \ell^{{d-1}}}^{\frac 12}  (\ell+1)^d \e^{3L^\beta} \e^{-m_1^\pr h_I (\lambda){ {R_{ \Theta}(y)}}}\abs{\psi(v^\pr)} \le  \e^{4L^\beta} \e^{-m_1^\pr h_I (\lambda){ {R_{ \Theta}(y)}}}\abs{\psi(v^\pr)} ,
\end{align}
for some  $v \in \La_\ell \cup\partial_{\mathrm{ex}}^{{\Theta}}  \Lambda_{\ell}$, and then for some 
$v^\pr \in\partial_{\mathrm{ex}}^{{\Theta}} \Lambda_{\ell}$ using Lemma~\ref{outbad}.

Combining \eq{Jdecomp}, \eq{barJest} and \eq{Jest32} yields
\begin{align}\label{noJest}
\abs{\scal {\e^{-\tfrac{m_1^\pr {R_{ \Theta}(y)}}{{A}^2} \pa{H_{\La_\ell}^2-\lambda^2}} \delta_y,\psi}}&\le 2 \e^{4L^\beta}\e^{-m_1^\pr  h_{I}(\lambda) {R_{ \Theta}(y)}}\abs{\psi(v)},
\end{align}
for some  $v \in\partial_{\mathrm{ex}}^{{\Theta}}  \Lambda_{\ell}$.

We now use  \cite[Lemma~3.2(ii)]{EK2} (it follows from  \eq{upbm}  that $\frac {\ell^{-\kappa^\pr} } 2<m_1^\pr \le  m \le \frac 1 2 \log \pa{1 + \tfrac {A}{4d}}$), getting
\begin{align}\label{eq:anpar99}
&\abs{\scal {F_{\frac {m_1^\pr {R_{ \Theta}(y)}}{{A}^2},\lambda} (H_\La)\delta_y,\Gamma_{ \boldsymbol{ \partial}^{ \Theta}  \La_\ell} \psi} }\le     70 s_d \ell^{d-1}A^{-1}\e^{-m_1^\pr  h_I(\lambda)  {R_{ \Theta}(y)}}\abs{\psi(v)},
\end{align}
for some  $v \in {\partial}_{\mathrm{ex}}^{\Th} \La_\ell$.   
We conclude  from \eq{noJest}  and   \eq{eq:anpar99} that
\begin{align}\label{decayest9999} 
&\abs{\psi(y)}\le C_{d}\pa{ \ell^{d-1 + \kappa^\pr} + \e^{4L^\beta}}\e^{-m_1^\pr h_{I}(\lambda) {R_{ \Theta}(y)}}\abs{\psi(v)}\\ \nn & \quad \le C_d^\pr \e^{4L^\beta}\e^{-m_1^\pr h_{I}(\lambda) {R_{ \Theta}(y)}}\abs{\psi(v)} \le  \e^{-m_3 h_{ I}(\lambda)  {R_{ \Theta}(y)}}  \abs{\psi(v^\pr)}\sqtx{for some} v \in\partial_{\mathrm{ex}}^{{\Theta}} \Lambda_{\ell},
\end{align}
where, using $h_{ I}(\lambda)  \ge \ell^{-\kappa}$ since  $\lambda \in I_{\ell}$, we have 
\begin{align}
m_3 &\ge  m \pa{1 - C_{d}\ell^{-\min \set{\ttau -  \gamma \beta -\kappa-\kappa^\pr, \frac{1- \tau}2}}} = m \pa{1 - C_{d}\ell^{-\frac{1- \tau}2}}.
\end{align}
\end{proof}

\subsection{Buffered subsets} 

The probability estimates of a multiscale analysis do not allow  all boxes to be be localizing, so we must control non-localizing boxes.
If a box $\La_\ell \subset \La_L $ is not $(m,I)$-localizing for $H$, we will add a  buffer of $(m,I)$-localizing boxes and study eigensystems for the enlarged subset.

  \begin{definition}\label{defbuff} 
 We call $\Ups\subset \La_L$ an $(m,I)$-buffered subset of the box  $\La_L$   if 
 the following holds:
  
  \begin{enumerate}
\item  $\Ups $ is a connected set in $\Z^d$ of the form
\beq\label{defUpsinitial}
 \Ups= \bigcup_{j=1}^J \La_{R_j}(a_j)\cap \La_L,
 \eeq  
 where $J\in \N$, $a_1,a_2,\ldots, a_J \in \La^\R_L$, and $\ell \le R_j\le L$ for $j=1,2,\ldots,J$.

  \item  There exists $\cG_\Ups \subset \La^\R_L$ such that:
  \begin{enumerate}
\item $\Lambda_\ell(a)\subset \La_L$ for all $a \in \cG_\Ups$ and  $\set{\Lambda_\ell(a)}_{a\in \cG_\Ups} $ is  a collection of  $(m,I)$-localizing boxes for $H$.

\item  For all $y \in \partialin^{\La_L}\Ups$ there exists $a_y \in\cG_\Ups$ such that $y\in \La_\ell^{ {\La_L}, {\ell_{\ttau}}}(a_y)$.
\end{enumerate}

  \end{enumerate}   
\end{definition}

This definition of a buffered subset  has subtle but   important differences from \cite[Definition~3.6]{EK2}, 
in addition to not requiring   level spacing conditions.   Definition~\ref{defbuff}(ii)   requires $\Lambda_\ell(a)\subset \La_L$  and $y\in \La_\ell^{ {\La_L}, {\ell_{\ttau}}}(a_y)$, while  the corresponding \cite[Definition~3.6]{EK2}(iii) has $\Lambda_\ell(a)\subset \Ups$ and  $y\in \La_\ell^{ {\Upsilon}, {2\ell_{\tau}}}(a_y)$.

 In the multiscale analysis we control the effect of buffered subsets using the following lemma.

\begin{lemma}\label{buffcontrol} Let $\La_L=\La_L(x_0)$, $x_0 \in \R^d$, and let  $(\psi,\lambda)$ be an  eigenpair for $H_{\La_L}$ with $\lambda \in I_\ell$. Let 
 $\Ups\subsetneq \La_L$ be an $(m,I)$-buffered subset,  and suppose
 \beq\label{distUps}
\dist\pa{\lambda, \sigma_I(H_{\Ups})} \ge \tfrac 12\e^{-L^{\beta}}\qtx{and} \min_{a\in\cG_\Ups} \dist\pa{\lambda, \sigma_I(H_{\Lambda_\ell(a)})} \ge \tfrac 12\e^{-L^{\beta}}.
\eeq
Then for all $y\in \Ups$ we have
\begin{align}\label{badboxest}
\abs{\psi(y) } \le \e^{-\frac {m_3}2  h_{I}\pa{\lambda}\ell_{\ttau}} \abs{\psi(y_1) } \qtx{for some}  y_1 \in \bigcup_{a\in \cG_\Ups} \partial_{\mathrm{ex}}^{ \La_L} \Lambda_{\ell}(a),
\end{align}
where $m_3= m_3(\ell)$ is as in \eq{m4}.
\end{lemma}

\begin{proof}  Let $y\in \Ups$.  In view of \eq{distUps} it follows from Lemma~\ref{outbad} that
\beq\label{distLaL3}
\abs{\psi(y) } \le 4d  \e^{L^{\beta}} \abs{ \partial_{\mathrm{ex}}^{ \La_L}{\Ups}}\abs{\psi(y_1) } \qtx{for some}  y_1\in  \partial_{\mathrm{ex}}^{ \La_L}{\Ups}.
\eeq
Let $a_1\in \cG_\Ups$ be such that $y_1 \in \La_\ell^{ {\La_L}, {\ell_{\ttau}}}(a_1)$. It then follows from \eq{distUps} and \eq{decayest12} in Lemma~\ref{lemdecay2} that
\beq  
\abs{\psi(y_1)}\le  \e^{-m_3 h_{I}\pa{\lambda} \ell_{\ttau}} \abs{\psi(y_2)} \qtx{for some} y_2 \in \partial_{\mathrm{ex}}^{ \La_L}\Lambda_{\ell}(a_1).
\eeq

Since $\abs{ \Ups}\le  \abs{ \La_L}\le (L+1)^d$ and
$\abs{ \partial_{\mathrm{ex}}^{ \La_L}{\Ups}}\le  2d \abs{ \Ups}\le 2d  (L+1)^d$, and we have \eq{lowerbdh} as $\lambda \in I_\ell$,
we get
\begin{align}\label{badboxest5}
\abs{\psi(y) } \le 8d^2 (L+1)^{d}\e^{L^{\beta}} \e^{-m_3  h_{I}\pa{\lambda}\ell_{\ttau}}  \abs{\psi(y_3) }\le  \e^{-\frac {m_3}2  h_{I}\pa{\lambda}\ell_{\ttau}},
\end{align}
for some  $y_3 \in \bigcup_{a\in \cG_\Ups} \partial_{\mathrm{ex}}^{ \La_L} \Lambda_{\ell}(a)$,
if $L$ is sufficiently large.
\end{proof}

  \section{Spectral separation}\label{secprobest}

We recall the Wegner estimate for the Anderson model as in  Definition~\ref{defineAnd} (see, e.g., 
\cite[Appendix~A]{CGK2}).

\begin{lemma}\label{lemWeg}  Let  $H_{\bom}$ be an Anderson model.
Let $\Th \subset \Z^d$. Then, for all  $E\in\R$,
\begin{align}
\P\set{\dist\set{E,\sigma (H_{\Theta,\bom})} \le \eta}\le   \wtilde{K}\eta^\alpha \abs{\Th},
\end{align}
where
with $\wtilde{K}=2K$ if $\alpha=1$ and $\wtilde{K}=8\, 2^\alpha K$ if  $\alpha \in (0,1)$.
\end{lemma}

\begin{definition}  Let $R>0$.   Two finite sets $\Theta, \Theta^\pr \subset \Z^d$ will be called  
$R$-separated  for $H$ if  $\dist \set{\sigma(H_{\Theta}),\sigma(H_{\Theta^\pr})} \ge \e^{-{R}^\beta}$, i.e.,  $\abs{\lambda- \lambda^\pr}\ge \e^{-{R}^\beta}$ for all $\lambda \in \sigma(H_{\Theta})$ and $ \lambda^\pr\in \sigma(H_{\Theta^\pr})$.
\end{definition}

\begin{definition}  Let $\Theta \subset \Z^d$ and $R>0$.  A family  $\set{\Phi_j}_{j\in J} $  of finite subsets of $\Th$ is called $R$-separated for $H$ if $\Phi_j$ and $\Phi_{j^\pr}$ are $R$-separated  for $H$ for all $j,j^\pr \in J$ such that  $\Phi_j\cap \Phi_{j^\pr}=\emptyset$.
\end{definition}

Lemma~\ref{lemWeg} implies   the Wegner estimate for $R$-separated sets
(see, e.g., \cite[Lemma~5.28]{Ki}).

\begin{lemma}\label{lemSep}  Let  $H_{\bom}$ be an Anderson model.
Let $\Theta, \Theta^\pr \subset \Z^d$ with $\Theta\cap \Theta^\pr =\emptyset$. Then, for all $0< \eta$,
\begin{align}
\P\set{\dist \set{\sigma(H_{\Theta}),\sigma(H_{\Theta^\pr})} \le \eta}\le \wtilde{K} \eta^\alpha\abs{\Th}\abs{\Th^\pr}.
\end{align}
In particular,
\begin{align}
\P\set{\Theta, \Theta^\pr \sqtx{are} R\text{-separated for}\  H }\ge 1 -\wtilde{K} \e^{-\alpha{R}^\beta}\abs{\Th}\abs{\Th^\pr}.
\end{align}
\end{lemma}

\section{Eigensystem multiscale analysis}\label{secEMSA}
In this section we fix an   Anderson model $H_{\bom}$ and prove Theorem~\ref{thmMSA}.

The following is an
 extension of Definition~\ref{defmIloc}.

\begin{definition}  \label{defmIloc9} 
Let  $ J=I(E,B) \subset I=I(E,A)$  be bounded open intervals with the same center.  A box $\La_L$ will be called $(m,J,I)$-localizing for $H$ if 
 \beq\label{upbmB}
L^{-\kappa^\pr} \le m   \le  \tfrac 1 2 \log \pa{1 + \tfrac {B}{4d}},
\eeq
 and
 there exists an  $(m,J,I)$-localized eigensystem for $H_{\La_L}$, that is,   an  eigensystem   $\set{(\vphi_\nu, \nu)}_{\nu \in \wtilde\sigma(H_{\La_L})}$ for $H_{\La_L}$ such that for all $\nu \in \wtilde\sigma(H_{\La_L})$ there is $x_\nu\in \La_L$ so $\vphi_\nu$ is $(x_\nu, m \Chi_{J}(\nu)h_{ I}(\nu))$-localized.
 \end{definition}
 
 Note that   $(m,I,I)$-localizing/localized   is the same as   $(m,I)$-localizing/localized.
 If $\La_L$ is $(m,J,I)$-localizing for $H$ it is also
$(m,J)$-localizing for $H$ as  $\Chi_{J}h_{ I} \ge  h_{ J}$.

 \begin{proposition}\label{propMSA}   There exists a 
 a finite scale   $\cL= \cL(d) $ with the following property:  Suppose for some scale 
$L_0 \ge \cL$ and interval $I_0=I(E,{A_0})$ we have
  \begin{align}\label{initialconinduc}
\inf_{x\in \R^d} \P\set{\La_{L_0} (x) \sqtx{is}  (m_0,I_0) \text{-localizing for} \; H_{\bom}} \ge 1 -  \e^{-L_0^\zeta}.
\end{align}
Set
 $L_{k+1}=L_k^\gamma$,  $A_{k+1}= A_k (1- L_k^{-\kappa})$, and $I_{k+1}= I(E,A_{k+1})$,   for $k=0,1,\ldots$. 
 Then  for all $k=1,2,\ldots$ 
  we have
   \begin{align} \label{MSALk}
\inf_{x\in \R^d} \P\set{\La_{L_k} (x) \sqtx{is} ( m_k , I_k,I_{k-1})  \text{-localizing for} \; H_{\bom}} \ge 1 -  \e^{-L_k^\zeta} ,
\end{align}
where 
\begin{gather}\label{Minduc2}
  L_k^{-\kappa^\pr}  <  m_{k-1}\pa{1- C_{d}  L_{k-1}^{-\vrho}}
 \le m_k < \tfrac 1 2 \log \pa{1 + \tfrac {A_k}{4d}}.  \end{gather}
 \end{proposition}

The proof of Proposition~\ref{propMSA} relies on the following lemma, the induction step for the multiscale analysis.

\begin{lemma}\label{lemInduction}
 Let $I=(E,{A})$. Suppose for some scale $\ell$ we have
 \begin{align}\label{hypMSAlem}
\inf_{x\in \R^d} \P\set{\La_{\ell} (x) \sqtx{is}  (m,I) \text{-localizing for} \; H_{\bom}} \ge 1 -  \e^{-\ell^\zeta}.
\end{align}
Then, if $\ell$ is sufficiently large, we have (recall $L=\ell^\gamma$)
 \begin{align}
\inf_{x\in \R^d} \P\set{\La_{L} (x) \sqtx{is}   (M,I_\ell,I) \text{-localizing for} \; H_{\bom}} \ge 1 -  \e^{-L^\zeta},
\end{align}
where
 \begin{align}\label{Minduc}
 L^{-\kappa^\pr}  <  m\pa{1- C_{d} \ell^{-\vrho}}
 \le M < \tfrac 1 2 \log \pa{1 + \tfrac {A(1-\ell^{-\kappa})}{4d}}.
\end{align}
\end{lemma}

\begin{proof}
To prove the lemma we proceed as  in \cite[Proof of Lemma~4.2]{EK2},  with several modifications.

 We assume \eq{hypMSAlem} for a scale $\ell$. 
 We take
$\La_L=\La(x_0)$, where $x_0\in \R^d$,  and let
${\mathcal C}_{L,\ell}={\mathcal C}_{L,\ell} \left(x_0 \right)$ be the suitable $\ell$-cover of $\La_L$ with $\vs$ as in \eq{vsdef} (see Appendix~\ref{subsecsc}).  Given $a,b \in {\Xi}_{L,\ell}$, we will say that the boxes $\La_\ell(a)$ and $\La_\ell(b)$ are disjoint if and only if $ \La_\ell^\R(a) \cap \La_\ell^\R(b) =\emptyset$, that is, if and only if   $\norm{a-b} \ge k_\ell \rho \ell^\vs$  (see Remark~\ref{remdisj}).
  We take  (recall \eq{ttauzeta2})
  \beq\label{setN}
  N= N_\ell= \fl{\ell^{(\gamma-1)\tzeta}},
  \eeq and 
let  $\cB_N$ denote the event that there exist at most $N$ disjoint boxes in ${\mathcal C}_{L,\ell}$ that are not $(m,I)$-localizing for  $ H_{\bom}$. For sufficiently large $\ell$, we have, using \eq{number},  \eq{hypMSAlem}, and the fact that events on disjoint boxes are independent, that
\begin{align}\label{probBN}
\P\set{\cB_N^c}\le \pa{\tfrac{2L} {\ell^\vs}}^{(N+1)d} \e^{-(N+1)\ell^\zeta}= 2^{(N+1)d} \ell^{(\gamma-\vs)(N+1)d}\e^{-(N+1)\ell^\zeta} < \tfrac 12 \e^{-L^\zeta}.
\end{align}

We  now fix   $\bom \in\cB_N$.  There exists $\cA_N=\cA_N(\bom)\subset  \Xi_{L,\ell}=\Xi_{L,\ell} \left(x_0 \right)$ such that
$\abs{\cA_N}\le N$ and   $\norm{a-b} \ge k_\ell \rho \ell^\vs$ if $a,b \in \cA_N$ and $a\ne b$, with the following property:  if $a\in \Xi_{L,\ell} $ with
$\dist (a,\cA_N)\ge  k_\ell {\rho} \ell^\vs$, so $\La_\ell^\R(a)\cap \La_\ell^\R(b)=\emptyset$ for all $b\in \cA_N$,  the box $\La_\ell(a)$ is $(m,I)$-localizing for  $ H_{\bom}$.
 In other words,
\beq\label{implyloc}
a \in  \Xi_{L,\ell} \setminus \bigcup_{b\in \cA_N}\La^\R_{2(k_\ell-1) \rho \ell^\vs}(b) \quad \Longrightarrow  \quad \La_\ell(a) \sqtx{is} (m,I)\text{-localizing for} \quad  H_{\bom}.
\eeq

We  want to  embed the   boxes $\set{\La_\ell(b) }_{b \in \cA_{N}}$  into   $(m,I)$-buffered subsets of $\La_L$. To do so, we consider  graphs  $\G_i=\pa{ \Xi_{L,\ell}, \E_i}$, $i=1,2$,   
both having   $ \Xi_{L,\ell}$ as the set of vertices, with sets of edges given by
\begin{align}
 \E_1&=  \set{\set{a,b}\in  \Xi_{L,\ell}^2;\; 0<  \norm{a-b} \le (k_\ell-1)\rho \ell^\vs }  \\  \notag
 & =\set{\set{a,b}\in  \Xi_{L,\ell}^2;\;  a\ne b \sqtx{and} \La^\R_\ell(a)\cap \La^\R_\ell(b)\ne \emptyset },\\
 \notag
 \E_2&= \set{\set{a,b}\in  \Xi_{L,\ell}^2;\; k_\ell\rho \ell^\vs\le  \norm{a-b}\le (3k_\ell-1)\rho \ell^\vs}\\
 \notag &  =\set{\set{a,b}\in  \Xi_{L,\ell}^2;\; \La^\R_\ell(a)\cap \La^\R_\ell(b)=\emptyset\sqtx{and} \  \La^\R_{2k_\ell\rho \ell^\vs+\ell}(a)\cap \La^\R_{2k_\ell\rho \ell^\vs+\ell}(b)\ne \emptyset}.
 \end{align}
 
Given $\Psi \subset \Xi_{L,\ell} $, we  let   $\overline{\Psi}= \Psi \cup \partial_{\mathrm{ex}}^{ \G_1} \Psi$, where $ \partial_{\mathrm{ex}}^{ \G_1} \Psi$,   the exterior boundary of $\Psi$ in the graph $\G_1 $, is defined by
\begin{align}
 \partial_{\mathrm{ex}}^{ \G_1} \Psi&= \set{a\in \Xi_{L,\ell}\setminus \Psi; \ \dist (a,\Psi)\le (k_\ell-1)\rho \ell^\vs }\\  \nn  
 &= \set{a\in \Xi_{L,\ell}\setminus \Psi; \ (b,a) \in \E_1 \sqtx{for some} b \in \Psi } .
 \end{align}

Let  $\Phi\subset  \Xi_{L,\ell}$ be $\G_2$-connected, so
  $\diam \Phi \le  (3k_\ell-1)\rho \ell^\vs\pa{\abs{\Phi}-1}$. (The diameter of a set $\Xi\subset \R^d$ is given by $\diam \Xi= \sup_{x,y \in \Xi} \norm{y-x}$.)
 Then
 \beq
\wtilde{\Phi}= 
 \set{a\in \Xi_{L,\ell}; \; \dist (a,\Phi)\le k_\ell\rho \ell^\vs }
\eeq  
 is a  $\G_1$-connected subset of $ \Xi_{L,\ell}$ such that   
 \beq\label{diamwphi}
 \diam \wtilde{\Phi} \le  \diam {\Phi} +2k_\ell\rho \ell^\vs  \le  \pa{(3k_\ell-1)\abs{\Phi} -   (k_\ell-1)}\rho \ell^\vs\le 5\ell \abs{\Phi} . 
 \eeq 
We set 
\begin{align}\label{constUps}
\Ups_\Phi  = \bigcup_{a \in \wtilde{\Phi}}  \La_\ell (a) \qtx{and} \cG_{\Ups_\Phi }=\partial_{\mathrm{ex}}^{ \G_1}\wtilde{\Phi}.
\end{align}

Let $\set{\Phi_r}_{r=1}^R=\set{\Phi_r(\bom)}_{r=1}^R$ denote the $\G_2$-connected components of  $\cA_{N}$ (i.e., connected in the graph $\G_2$); we have 
$R \in \set{1,2,\ldots,N}$ and  $\sum_{r=1}^R\abs{\Phi_r}=\abs{\cA_N}\le N$.
We conclude  that   $\set{\wtilde{\Phi}_r}_{r=1}^R$ is a collection of disjoint, $\G_1$-connected subsets of $ \Xi_{L,\ell}$, such that
\begin{gather}
  \dist (\wtilde{\Phi}_r,\wtilde{\Phi}_s)\ge  k_\ell\rho \ell^\vs > \ell \qtx{if} r\ne s. \label{distPhi}
 \end{gather}

  Moreover, it follows from \eq{implyloc} that
\beq
\label{implyloc2}
a \in \cG=\cG(\bom)=  \Xi_{L,\ell} \setminus \bigcup_{r=1}^R\wtilde{\Phi}_r\quad \Longrightarrow  \quad \La_\ell(a) \sqtx{is} (m,I)\text{-localizing for} \quad  H_{\bom}.
\eeq
In particular, we conclude that $\La_\ell(a)$ is $(m,I)$-localizing for $H_{\bom}$ for all $a \in \partial_{\mathrm{ex}}^{ \G_1}\wtilde{\Phi}_r$, $r=1,2,\ldots, R$.

Each $\Ups_r=\Ups_{\Phi_r}$, $r=1,2,\dots,R$,  clearly satisfies all the requirements to be an $(m,I)$-buffered subset of $\La_L$ with $\cG_{\Ups_r}=\partial_{\mathrm{ex}}^{ \G_1}\wtilde{\Phi}_r$  (see Definition~\ref{defbuff}). 
Moreover the sets $\{\Ups_r\}_{r=1}^R$ are disjoint.  Note also that it follows from \eq{diamwphi} that  
\beq\label{diamUps}
\diam \Ups_r \le \diam {\wtilde{\Phi}}_r  + \ell  \le5\ell \abs{\Phi_r}  +\ell  \le   6\ell \abs{\Phi_r},
\eeq 
so, using  \eq{gamtzetabeta}, we have 
\beq\label{sumdiam}
\sum_{r=1}^R \diam \Ups_r\le 6\ell N  \le 6 \ell^{(\gamma-1)\tzeta  +1} \ll \ell^{\gamma \tau}=L^\tau.
\eeq

Let
\beq
\mathbb{S}_{\bom}= \set{\La_\ell(a)}_{a\in \cG} \cup\set{\Ups_r}_{r=1}^R.
\eeq
We can  arrange for $\mathbb{S}_{\bom}$  to be an $L$-separated family of subsets of $\La_L$ for $H$ as follows.
 Let
\begin{align}
\cF_N=\bigcup_{r=1}^N  \cF(r), \sqtx{where} \cF(r)=\set{\Phi \subset \Xi_{L,\ell} ; \;  \Phi \sqtx{is}  \G_2\text{-connected}\sqtx{and} \abs{\Phi}=r}.
\end{align}
We set $\wtilde{\mathbb{S}}_N= \set{\La_\ell(a)}_{a\in \Xi_{L,\ell}} \cup\set{\Ups_\Phi}_{\Phi \in \cF_N}$.  Given $S_1,S_2\in \wtilde{\mathbb{S}}_N$, $S_1\cap S_2=\emptyset$, it follows from Lemma~\ref{lemSep}  that 
\begin{align}\label{probspaced}
\P\set{S_1 \sqtx{and} S_2 \sqtx{are not} L\text{-separated for }H_{\eps,\bom} }\le \wtilde{K} e^{-\alpha{L}^\beta}\pa{L+1}^{2d}\le e^{-\frac \alpha 2 {L}^\beta}.
\end{align} 
We have $\abs{ \Xi_{L,\ell} }\le 2^d \ell^{(\gamma-\vs)d}$ from \eq{number}.
Setting $ \cF(r,a)=\set{\Phi \in \cF(r); \, a\in \Phi}$ for $a\in \Xi_{L,\ell}$, and letting
$\kappa(a)$ denote the number of nearest neighbors of $a\in  \Xi_{L,\ell} $ in the graph  $\G_2$,
and noting that
\begin{align}
 \kappa(a) & \le  \pa{2(3k_\ell-1)+1}^d- \pa{2(3k_\ell-2)+1}^d\le  d \pa{2(3k_\ell-1)+1}^{d-1}\\ \nn & =  d \pa{6k_\ell-1}^{d-1}\le d 20^{d-1} \ell^{(1-\vs)(d-1)}\le \ell^{d-1},
\end{align} 
 we get
\begin{align}\label{cFN}
\abs{\cF(r,a)}\le  (r-1)! \ell^{(d-1)(r-1)}\quad & \Longrightarrow \quad  \abs{\cF(r)}\le (L+1)^d (r-1)! \ell^{(d-1)(r-1)}\\ \notag
& \Longrightarrow \quad \abs{\cF_N}\le (L+1)^d N!\ell^{(d-1)(N-1)} .
\end{align}
Thus, we get 
\beq\abs{\wtilde{\mathbb{S}}}_N\le  2^d \ell^{(\gamma-\vs)d}+ (L+1)^d N! \ell^{(d-1)(N-1)}
\le 2( L+1)^d N! \ell^{(d-1)(N-1)}.
\eeq

Letting  $\cS_N$ denote that the event  that $\wtilde{\mathbb{S}}_N$ is an $L$-separated family of subsets of $\La_L$ for $H$,   and taking  $N=N_\ell$ as  in \eq{setN}, we get 
\beq\label{probSN}
\P\set{\cS_N^c} \le e^{-\frac \alpha 2 {L}^\beta}2( L+1)^d N_\ell! \ell^{(d-1)(N_\ell-1)} <e^{-\frac \alpha 4 {L}^\beta} < \tfrac 12 \e^{-L^\zeta} ,
\eeq
for sufficiently large $L$, since $(\gamma-1)\tzeta < (\gamma-1)\beta<\gamma \beta$ and $\zeta < \beta$.

We now define the event  $\cE_N= \cB_N \cap \cS_N$.  
 It follows from \eq{probBN} and \eq{probSN} that
\beq
\P\set{\cE_N}> 1-  \e^{-L^\zeta}.
\eeq
To finish the proof we need to show that for all  $\bom \in \cE_N$ the box $\La_L$ is 
$   (M,I_\ell,I)$-localizing for $ H_{\bom}$, where $M$ is given  in \eq{Minduc}.

Let us fix $\bom \in \cE_N$.   Then we have \eq{implyloc2},  the subsets $\set{\Ups_r}_{r=1}^R$ constructed in \eq{constUps} are buffered subsets of $\La_L$ for $ H_{\bom}$, and  the collection $\mathbb{S}_{\bom}$ is  an $L$-separated family of subsets of $\La_L$ for $H$.
It follows from \eq{covproperty} and Definition~\ref{defbuff}(ii) that
\beq\label{LadecompU}
\La_L=\set{ \bigcup_{a \in  \cG} {\Lambda}_{\ell}^{\La_L, \frac {\ell -\ell^\vs}2}(a)}\cup \set{\bigcup_{r=1}^R\Ups_r}.
\eeq 
Note that ${\Lambda}_{\ell}^{\La_L, \frac {\ell -\ell^\vs}2}(a) \subset  {\Lambda}_{\ell}^{\La_L, \ell_{\ttau}}(a)$.

 Let   $\set{(\psi_\lambda, \lambda)}_{\lambda \in \wtilde\sigma(H_{\La_L})}$ be an eigensystem   for $H_{\La_L}$. ( Since $\bom$ is  fixed, we omit it from the notation.)   Given $\lambda \in \wtilde\sigma_{I_\ell}(H_{\La_L})$,  we claim there exists  $S_\lambda \in \mathbb{S}_{\bom}$ such that  
 \beq\label{distSlamb}
 \dist \pa{\lambda,  \sigma(H_{S_\lambda})}\le \tfrac 1 2\e^{-L^\beta}.
 \eeq
 
 Suppose not, i.e., $\dist \pa{\lambda,  \sigma(H_{S})}>\tfrac 1 2\e^{-L^\beta}$ for all $S\in \mathbb{S}_{\bom}$.  Let $y\in \La_L$.
 If $y \in {\Lambda}_{\ell}^{\La_L, \frac {\ell -\ell^\vs}2}(a)$ for some $a\in \cG$, we have ${R_{ \La_L}(y)}\ge \fl{\frac {\ell -\ell^\vs}2}$, so it follows from \eq{decayest12} that  
 \beq
 \abs{\psi_{\lambda}(y)}\le  \e^{-m_3 h_{I}\pa{\lambda} \fl{\frac {\ell -\ell^\vs}2}}\le  \e^{-m_3 \ell^{-\kappa} \fl{\frac {\ell -\ell^\vs}2}}\le  \e^{-\frac 1 4 m_3 \ell^{1-\kappa}}.
 \eeq 
 If not, it follows from \eq{LadecompU} that $y \in \Ups_r$ for some $r\in \set{1,2\ldots, R}$. But then it follows from \eq{badboxest} in Lemma~\ref{buffcontrol}  that
\beq
\abs{\psi_{\lambda}(y) } \le  \e^{-\frac {m_3} 2 h_{I}\pa{\lambda}\ell_{\ttau}} \le  \e^{-\frac {m_3} 2 \ell^{-\kappa}\ell_{\ttau}} \le  \e^{-\frac 1 4 m_3 \ell^{\ttau-\kappa}}.
\eeq 
 We conclude that
 \beq
 1=\norm{\psi_{\lambda}}^2 \le (L+1)^d   \e^{-\frac 1 4 m_3 \ell^{\ttau-\kappa}} <1,
 \eeq
 a contradiction.

We now pick $x_\lambda \in S_\lambda$.  We will show that $\psi_\lambda$ is an $(x_\lambda, M h_{ I}(\lambda))$-localized eigenfunction  for
 $H_{\bom}$, where $M$ is given  in \eq{Minduc}.  
  
  Let $ \mathbb{S}_{\bom} \up{\lambda}=\set{S\in \mathbb{S}_{\bom}; \ S\cap S_\lambda = \emptyset}$. If $S\in  \mathbb{S}_{\bom} \up{\lambda}$, $S$ and $S_\lambda$ are $L$-separated, so  it follows from \eq{distSlamb} that
 \beq
 \dist \pa{\lambda,  \sigma(H_{S})}\ge \dist \pa{\sigma(H_{S}),  \sigma(H_{S_\lambda})}- \dist \pa{\lambda,  \sigma(H_{S_\lambda})}\ge \e^{-L^\beta}- \tfrac 1 2\e^{-L^\beta}=\tfrac 1 2\e^{-L^\beta}.
 \eeq
 We consider two cases:

\begin{enumerate}
\item Let $y \in {\Lambda}_{\ell}^{\La_L, \frac {\ell -\ell^\vs}2}(a)$, where  ${\Lambda}_{\ell}(a) \in  \mathbb{S}_{\bom} \up{\lambda}$. In this case it follows from  \eq{decayest12} that 
\beq\label{psidecgoodpr2} 
\abs{\psi_{\lambda}(y)}\le  \e^{-m_3 h_{I}\pa{\lambda} \fl{\frac {\ell -\ell^\vs}2}}  \abs{\psi_{\lambda}(y_1)} \sqtx{for some} y_1\in\partial_{\mathrm{ex}}^{\La_L} \Lambda_{\ell}(a), 
\eeq
where $m_3= m_3(\ell)$ is as in \eq{m4}.  Moreover, we have
\beq\label{psidecgoodpr299} 
\norm{y-y_1}\le \ell +1 -  \fl{\tfrac {\ell -\ell^\vs}2} \le \tfrac {\ell +\ell^\vs}2 + 2\le  \tfrac {\ell +2\ell^\vs}2.
\eeq

\item  Let $y \in \Ups_r$, where $\Ups_r\in \mathbb{S}_{\bom} \up{\lambda}$ and
$\set{{\Lambda}_{\ell}(a) }_{a\in \cG_{\Ups_r}}\subset \mathbb{S}_{\bom} \up{\lambda}$. Then
it follows from \eq{badboxest} in Lemma~\ref{buffcontrol}  that
\begin{align}\label{gggsum356}
\abs{\psi_{\lambda}(y) } \le \e^{-\frac {m_3}2  h_{I}\pa{\lambda}\ell_{\ttau}} \abs{\psi_{\lambda}(y_2) }\le  \e^{-\frac  {m_3}4 \ell^{\ttau-\kappa}}\abs{\psi_{\lambda}(y_2) } \end{align}
for some $ y_2 \in \bigcup_{a\in \cG_{\Ups_r}}\partial_{\mathrm{ex}}^{{\La_L}} \Lambda_{\ell}(a)$,
where $m_3= m_3(\ell)$ is as in \eq{m4}. Note  that
\beq\label{gggsum35699}
\norm{y-y_2}\le \diam \Ups_r + \ell.
\eeq

\end{enumerate}

Now let us take $y\in  \La_L $ such that $\norm{y-x_\lambda} \ge {L_\tau}$.
Suppose $\abs{\psi_\lambda(y)} >0$, since otherwise there is nothing to prove.  We estimate $\abs{\psi_\lambda(y)} $ using either \eq{psidecgoodpr2}  or  \eq{gggsum356} repeatedly, as appropriate, stopping when we get too close  to $x_\lambda$ so we are not in one of the two cases described above.  (Note that this must happen since $\abs{\psi_\lambda(y)} >0$.) We accumulate decay only when we use \eq{psidecgoodpr2}, and just use $ \e^{-\frac {m_3 }4\ell^{\ttau-\kappa}}< 1$ when using \eq{gggsum356}. In view of \eq{psidecgoodpr299}  and \eq{gggsum35699}, this can be done  using \eq{psidecgoodpr2} at least
$S$ times, as long as 
\beq
   \tfrac {\ell +2\ell^\vs}2 S +\sum_{r=1}^R \pa {\diam \Ups_r + \ell}+ 2\ell  \le \norm{y-x_\lambda}.
\eeq
Since $\sum_{r=1}^R \pa {\diam \Ups_r + \ell}\le 7\ell N$ in view of  \eq{sumdiam}, this can be guaranteed by requiring
\begin{align}
    \tfrac {\ell +2\ell^\vs}2 S +  7 \ell^{(\gamma-1)\tzeta  +1} +2\ell  \le \norm{y-x_\lambda}.
\end{align}
We can thus have
\begin{align}
S& =  \fl{\tfrac 2{\ell +2\ell^\vs}\pa{\norm{y-x_\lambda}-7 \ell^{(\gamma-1)\tzeta  +1}  -2\ell} }- 1
 \\ \nn  & \ge  \tfrac 2{\ell +2\ell^\vs}\pa{\norm{y-x_\lambda}- 7 \ell^{(\gamma-1)\tzeta  +1}  -2\ell} -2  \\ \nn  & =  \tfrac 2{\ell +2\ell^\vs}\pa{\norm{y-x_\lambda}- 7 \ell^{(\gamma-1)\tzeta  +1}  -3\ell-2\ell^\vs}  \ge  \tfrac 2{\ell +2\ell^\vs}\pa{\norm{y-x_\lambda}- 8 \ell^{(\gamma-1)\tzeta  +1}  } .
\end{align}
Thus we conclude that 
\begin{align}\label{repeateddecay}
\abs{\psi_\lambda(y)} &\le   \e^{-m_3 h_{I}(\lambda) \fl{\frac {\ell -\ell^\vs}2} { \tfrac 2{\ell +2\ell^\vs}\pa{\norm{y-x_\lambda}- 8 \ell^{(\gamma-1)\tzeta  +1}  }
}}  \le    \e^{-M h_{I}(\lambda)\norm{y-x_\lambda}}
\end{align}
where 
\begin{align}
M&\ge m_3\pa{1- C_{d} \ell^{-\min\set{{1- \vs }, \gamma \tau- (\gamma-1)\tzeta  -1}}}\\ \nn & = m_3\pa{1- C_{d} \ell^{-\pa{ \gamma \tau- (\gamma-1)\tzeta  -1}}}
\\ \nn   &\ge 
m\pa{1- C_{d} \ell^{-\min\set{\kappa, \frac{1- \tau}2, \gamma \tau- (\gamma-1)\tzeta  -1}}}=m\pa{1- C_{d} \ell^{-\vrho}},
\end{align}
where we used \eq{vsdef}, \eq{m4}, and \eq{defvrho}.  In particular, $M$ satisfies \eq{Minduc} for sufficiently large $\ell$.

We conclude that $\psi_\lambda$ is an $(x_\lambda, M h_{ I}(\lambda))$-localized eigenfunction 
for $\La_L$, where $M$ satisfies \eq{Minduc}.

We proved that $\La_L$ is $(M,I_\ell,I)$-localized for $H_\bom$.
\end{proof}

\begin{proof}[Proof of Proposition~\ref{propMSA}]
We assume \eq{initialconinduc}  and set  $L_{k+1}=L_k^\gamma$,  $A_{k+1}= A_k (1- L_k^{-\kappa})$, and $I_{k+1}= I(E,A_{k+1})$   for $k=0,1,\ldots$. 
Since if a box  $\La_{L}$ is  $(M, I_\ell,I)$-localizing for $ H_{\bom}$ it is also $(M, I_\ell)$-localizing, 
if $L_0$ is sufficiently large it follows from Lemma~\ref{lemInduction} by an induction argument that
 we have \eq{MSALk} and \eq{Minduc2} for all $k=1,2, \ldots$.
\end{proof}

 \begin{proposition}\label{propMSAnok}  There exists a 
 a finite scale   $\cL= \cL(d) $ with the following property:  Suppose for some scale 
$L_0 \ge \cL$ and interval $I_0=I(E,{A_0})$ we have
  \begin{align}\label{initialconinduc993}
\inf_{x\in \R^d} \P\set{\La_{L_0} (x) \sqtx{is}  (m_0,I_0) \text{-localizing for} \; H_{\bom}} \ge 1 -  \e^{-L_0^\zeta}.
\end{align} 
Set
 $L_{k+1}=L_k^\gamma$,  $A_{k+1}= A_k (1- L_k^{-\kappa})$, and $I_{k+1}= I(E,A_{k+1})$,   for $k=0,1,\ldots$, 
 Then  for all $k=1,2,\ldots$   we have
   \begin{align} \label{MSALnok}
\inf_{x\in \R^d} \P\set{\La_{L} (x) \sqtx{is} ( m_k , I_{k},I_{k-1})  \text{-localizing for} \; H_{\bom}} \ge 1 -  \e^{-L^\xi}   \sqtx{for }     L\in [L_k, L_{k+1}),
\end{align}
where 
\begin{gather}\label{Minduc2333}
  L_k^{-\kappa^\pr}  <  m_{k-1}\pa{1- C_{d}  L_{k-1}^{-\vrho}}
 \le m_k < \tfrac 1 2 \log \pa{1 + \tfrac {A_k}{4d}}, \end{gather}
 with  $C_{d}$ as in \eq{Minduc2}.
 \end{proposition}

\begin{proof}  We  apply  Proposition~\ref{propMSA}, which gives a scale $\cL$ such that,
taking $L_0 \ge \cL$ 
 we have the conclusions of  Proposition~\ref{propMSA}.

Given a scale   $L\ge L_1$, let $k=k(L)\in \set{1,2,\ldots}$ be defined by
$L_k \le L <L_{k+1}$.  We have
 $L_k=L_{k-1}^\gamma\le L< L_{k+1}=L_{k-1}^{\gamma^2}$,  so $L=  L_{k-1}^{\gamma^\prime}$ with $\gamma \le \gamma^\prime<\gamma^2$. We proceed as in Lemma~\ref{lemInduction}. We take
$\La_L=\La_L(x_0)$, where $x_0\in \R^d$, let    $\set{(\psi_\lambda, \lambda)}_{\lambda \in \wtilde\sigma(H_{\La_L})}$ be an eigensystem   for $H_{\La_L}$,  and let
${\mathcal C}_{L,L_{k-1}}={\mathcal C}_{L,L_{k-1}} \left(x_0 \right)$ be the suitable $L_{k-1}$-cover of $\La_L$.  We let  $\cB_0$ denote the event that all  boxes in ${\mathcal C}_{L,L_{k-1}}$  are  $(m_{k-1},I_{k-1})$-localizing for  $ H_{\bom}$.  It follows from  \eq{number} and \eq{MSALk} that 
\begin{align}\label{probB0}
\P\set{\cB_0^c}\le \pa{\tfrac{2L} {L_{k-1}^\vs}}^{d} \e^{-L_{k-1}^\zeta}= 2^{d} L_{k-1}^{(\gamma^\pr-\vs)d}\e^{-L_{k-1}^\zeta}\le     2^{d} L^{(1-\frac\vs { \gamma^\pr})d}\e^{-L^{\frac \zeta { \gamma^\pr}}}  < \tfrac 12 \e^{-L^\xi} ,
\end{align}
if $L_0$ is sufficiently large, since $\xi \gamma^\pr <\xi \gamma^2 < \zeta$. 
Moreover, given $\La_1,\La_2\in {\mathcal C}_{L,L_{k-1}}$, $\La_1\cap\La_2=\emptyset$, it follows from Lemma~\ref{lemSep}  that 
\begin{align}\label{probspacedwww}
\P\set{\La_1 \sqtx{and} \La_2 \sqtx{are not} L\text{-separated for }H_{\bom} }\le \wtilde{K} e^{-\alpha{L}^\beta}\pa{L_{k-1}+1}^{2d}\le e^{-\frac \alpha 2 {L}^\beta}.
\end{align} 
Thus,
letting $\cS_0$ denote the event that ${\mathcal C}_{L,L_{k-1}}$ is an $L$-separated family of subsets of $\La_L$ for $H$, it follows from \eq{number} that
\begin{align}\label{probspacedLL}
\P\set{\cS_0^c }\le \pa{\tfrac {2L} {L_{k-1}^\vs}}^{2d} e^{-\frac \alpha 2 {L}^\beta}\le  \tfrac 12 \e^{-L^\xi},
\end{align}
if $L_0$ is sufficiently large, since $\xi < \beta$.  Thus, letting $\cE_0= \cB_0 \cap \cS_0$, we have 
\beq
\P\set{\cE_0 }\ge 1-  \e^{-L^\xi}.
\eeq

It only remains to prove that  $\La_L$ is  $(m_k,I_k,I_{k-1})$-localizing for  $ H_{\bom}$ for all $\bom \in \cE_0$.  To do so,  we fix $\bom \in \cE_0$ and  proceed as in the proof of Lemma~\ref{lemInduction}.   Since $\bom \in \cB_0$, $\La_{L_{k-1}} (a)$ is $(m_{k-1},I_{k-1})$-localizing for  $ H_{\bom}$  for all $a\in \cG=\Xi_{L,L_{k-1}}$. Since  $\bom$ is now fixed, we omit them from the notation. ,

Let $\lambda \in \wtilde\sigma_{I_k} (H_{\La_L})$ (note  $(I_{k-1})_{L_{k-1}}= I_{k}$).
To finish the proof we need to show  that $\psi_\lambda$ is    $  (m_k, I_k,I_{k-1})$-localized.
Since ${\mathcal C}_{L,L_{k-1}}$ is an $L$-separated family of subsets of $\La_L$ for $H$, there must exist  $a_\lambda \in \cG= \Xi_{L,L_{k-1}}$ such that, setting $\Lambda_\lambda= \La_{L_{k-1}}(a_\lambda)$, we have (as in the proof of Lemma~\ref{lemInduction})
\beq\label{distSlamb333}
 \dist \pa{\lambda,  \sigma(H_{\La_\lambda})}\le \tfrac 1 2\e^{-L^\beta},
 \eeq
and if $a \in \cG_\lambda= \set{b \in \cG; \ \La_{ L_{k-1}}(b)\cap \La_\lambda=\emptyset}$,
 \beq
 \dist \pa{\lambda,  \sigma(H_{\La})}\ge \tfrac 1 2\e^{-L^\beta}.
 \eeq
 If $y \in \La_L$ and $\norm{y-a_\lambda} \ge 2L_{k-1}$, it follows from \eq{covproperty} that $y \in  {\Lambda}_{L_{k-1}}^{\La_L, \frac {L_{k-1} -L_{k-1}^\vs}2}(a)$  for some $a\in \cG_\lambda$,  so it follows from \eq{decayest12} that
 \beq\label{psidecgoodpr256} 
\abs{\psi_\lambda(y)}\le \e^{-m_{k-1,3} h_{I_{k-1}}(\lambda) \fl{\frac {L_{k-1} -L_{k-1}^\vs}2}}  \abs{\psi_\lambda (y_1)},\eeq
for some $ y_1 \in{\partial}^{\La_L,2\pa{L_{k-1}}_{\tau }}  \Lambda_{L_{k-1}}(a)$,
where we need
\beq\label{m4k}
m_{k-1,3}= m_{k-1,3} (L_{k-1})\ge  m_{k-1} \pa{1 - C_{d}L_{k-1}^{-( \frac{1- \tau}2)}},
\eeq
 and we have  
 \beq \label{psidecgoodpr29945}
 \norm{y-y_1}\le \tfrac {L_{k-1}+2L_{k-1}^\vs}2,
 \eeq as in
 \eq{psidecgoodpr299}.

 Now consider  $y\in  \La_L $ such that $\norm{y-a_\lambda} \ge {L_\tau}$.
Suppose $\abs{\psi_\lambda(y)} >0$, since otherwise there is nothing to prove.  We estimate $\abs{\psi_\lambda(y)} $ using either \eq{psidecgoodpr256} repeatedly, as appropriate, stopping when we get within $2 L_{k-1}$ of $a_\lambda$.  In view of \eq{psidecgoodpr29945}  , we can use \eq{psidecgoodpr256} 
$S$ times, as long as 
\beq
  \tfrac {L_{k-1}+2L_{k-1}^\vs}2 S + 2L_{k-1} \le \norm{y-a_\lambda}.
\eeq 
We can thus have
\begin{align}\nn
S& =  \fl{\tfrac 2 {L_{k-1}+2L_{k-1}^\vs}\pa{\norm{y-a_\lambda}-2L_{k-1}} }- 1
 \ge  {\tfrac 2 {L_{k-1}+2L_{k-1}^\vs}\pa{\norm{y-a_\lambda}-2L_{k-1}} }-2  \\  & \ge \tfrac 2 {L_{k-1}+2L_{k-1}^\vs}\pa{\norm{y-a_\lambda}-3L_{k-1}-2L_{k-1}^\vs}  \ge  \tfrac 2 {L_{k-1}+2L_{k-1}^\vs}\pa{\norm{y-a_\lambda}- 4L_{k-1} } .
\end{align}

Thus we conclude that 
\begin{align}\label{repeateddecay333}
\abs{\psi_\lambda(y)} &\le   \e^{-m_{k-1,3} h_{I_{k-1}}(\lambda) \fl{\frac {L_{k-1} -L_{k-1}^\vs}2}  {\tfrac 2 {L_{k-1}+2L_{k-1}^\vs}}\pa{\norm{y-a_\lambda}- 4L_{k-1} }
}\\ \nn &  \le    \e^{-m_k h_{I_{k-1}}(\lambda)\norm{y-a_\lambda}}
\end{align}
where  $m_k$ can be taken to satisfy \eq{Minduc2}.

We conclude that $\psi_\lambda$ is an $(m_k,I_k,I_{k-1})$-localized eigenfunction, where $m_k$ satisfies \eq{Minduc2}.

 We proved that the box  $\La_L$ is 
$ (m_k, I_k,I_{k-1})$-localizing for $ H_{\bom}$.
\end{proof}

\begin{proof} [Proof of Theorem~\ref{thmMSA}]

 Let   $L_{k+1}=L_k^\gamma$,  $A_{k+1}= A_k (1- L_k^{-\kappa})$,  $I_{k+1}=I(E,A_{k+1})$, and $m_{k+1}= m_{k}\pa{1- C_{d}  L_{k}^{-\vrho}}
  $   for $k=0,1,\ldots$.  Given $ L\ge L_0^\gamma =L_1$,  let $k=k(L)\in \set{1,2,\ldots}$ be defined by
$L_k \le L <L_{k+1}$.   Let $A_\infty, I_\infty, m_\infty$ be defined by \eq{Aminfty}.   Since
\beq
A_k= A_\infty \prod_{j=k}^\infty \pa{1- L_j^{-\kappa}}^{-1}\qtx{for} k=0,1,\ldots,
\eeq
we have
\beq
A_\infty \pa{1- L^{-\kappa}}^{-1} \le A_\infty \pa{1- L_k^{-\kappa}}^{-1} < A_k,
\eeq
and hence  $I_\infty^L \subset I_k$.  Since $m_\infty\le m_k$, 
we conclude that
 \eq{MSALnok2} follows from \eq{MSALnok}.  
\end{proof}

\section{Localization}\label{seclocproof}
 In this section  we prove Theorem~\ref{thmloc} for an   Anderson model $H_{\bom}$.

 \begin{lemma}\label{lemWimp} Let   $I=(E,A)$.  There exists a finite scale  $\cL_{d,\nu}$ such that  
 for all $L\ge \cL_{d,\nu}$ and  $a\in \Z^d$,   given    an $(m,I^{L})$-localizing box $\La_L(a)$  for the discrete Schr\"odinger operator  $H$,  then for all $\lambda \in I$,  \beq
\max_{b\in \La_{\frac L 3}(a)} W\up{a}_{\lambda}(b)> \e^{-\frac 1 4 m h_{ I^{L}} (\lambda) L}\quad \Longrightarrow \quad  \min_{\theta\in \sigma_{ I^{L}}(H_{\La_L(a)})} \abs{\lambda -\theta}  < 
\tfrac 1 2 \e^{-L^{\gamma\beta}}.
\eeq 
 \end{lemma}
  
 \begin{proof}

Let $\lambda \in I =\pa{I^L}_L$, and  suppose $ \abs{\lambda -\theta}  \ge\tfrac 1 2\e^{-L^{\gamma\beta}}$ for all $\theta\in \sigma_{I^L}(H_{\La_L(a)})$.  Let $\psi \in \cV(\lambda)$.  Then it follows from Lemma~\ref{lemdecay2} that  for large $L$ and $b\in \La_{\frac L 3}(a)$  we have
 \beq
 \abs{\psi(b)}\le \e^{-m_3(L) h_{I^L}(\lambda)\pa{\frac L 3 -1} }\norm{T_a^{-1} \psi}
  \scal{\tfrac L 2 + 1}^{\nu} \le \e^{-\frac 1 4 m h_{I^L}(\lambda)L} \norm{T_a^{-1} \psi}.
 \eeq
 \end{proof} 
 
\begin{proof}[Proof of Theorem~\ref{thmloc}]  Assume  Theorem~\ref{thmMSA} holds for some $L_0$, 
and let $I=I_\infty$, $m=m_\infty$. 
Consider  $L_0^\gamma \le L \in 2\N$  and $a\in \Z^d$.  We have  
 \beq
  \La_{5L}(a) =\bigcup_{b\in \set{ a+ \frac 1 2L  \Z^d}, \ \norm{b-a}\le  2L} \La_{L}(b).
  \eeq
  Let $\cY_{L,a}$ denote the event that  $\set{ \La_{L}(b)}_{b\in \set{ a+ \frac 1 2L  \Z^d}, \ \norm{b-a}\le  2L}$ is an $L^\gamma$-separated  family of   $(m,I^L)$-localizing boxes for $H$.  It follows from \eq{MSALnok2} and Lemma~\ref{lemSep}  that
  \beq
  \P\set{\cY_{L,a}^c}\le 9^d \e^{-L^\xi}  +   \wtilde{K} 9^{2d}\pa{L+1}^{2d} \e^{-\alpha L^{\gamma\beta}}\le C_{\mu} \e^{-L^\xi} .
  \eeq
  
  Suppose $\bom \in \cY_{ L,a}$, $\lambda \in I$,  and $\max_{b\in \La_{\frac L 3}(a)} W\up{a}_{\bom,\lambda}(b)> \e^{-\frac 1 4 m h_{I^L} (\lambda)L }$. It follows from Lemma~\ref{lemWimp} that $ \min_{\theta\in \sigma_{I^L}(H_{\La_L(a)})} \abs{\lambda -\theta}  < \tfrac 1 2\e^{-L^{\gamma\beta}}$.  Since the family of boxes is $L^\gamma$-separated  family for $H_{\bom}$,   we conclude that 
  \begin{align} 
   \min_{\theta\in \sigma_{I^L}(H_{\La_L(b)})} \abs{\lambda -\theta} \ge \tfrac 1 2 \e^{-L^{\gamma\beta}}
\end{align}
  for all $b\in \set{ a+ \frac 1 2L  \Z^d}$ with $\frac 32 L\le\norm{b-a}\le 2L$. Since
  \beq\label{Aell2}
 A_L(a)\subset \bigcup_{b\in \set{ a+ \frac 1 2L  \Z^d}, \frac 32  L \le \norm{b-a}\le 2L} \La_{L}^{\frac L 7}(b),
  \eeq
 it follows from Lemma~\ref{lemdecay2} that for all  $y\in A_L(a)$  we have, given  $\psi \in \cV_{\bom}(\lambda)$,
 \begin{align}
 \abs{\psi(y)}&\le  \e^{-m_3(L)  h_{I^L} (\lambda)\pa{\frac L 7 - 2}} \norm{T_a^{-1} \psi}\la \tfrac 5 2 L  +1 \ra^\nu\le \e^{-m h_{I^L} (\lambda) \frac L 8}\norm{T_a^{-1} \psi}&  \\  \nn
 & \le  \e^{- \frac 7 {132} m   h_{I^L} (\lambda)\norm{y-a} }\norm{T_a^{-1} \psi},
 \end{align}
so we get
\beq
W\up{a}_{\bom,\lambda}(y)\le \e^{-\frac 7 {132}        m  h_{I^L} (\lambda)\norm{y-a}} \qtx{for all} y\in A_L(a).
\eeq

Since we have \eq{boundGW}, we conclude that for $\bom \in \cY_{L,a}$ we always have
\begin{align} 
W\up{a}_{\bom,\lambda}(a)W\up{a}_{\bom,\lambda}(y) & \le 
\max \set{\e^{- \frac 7 {66}  m  h_{I^L} (\lambda)\norm{y-a}}\la  y-a\ra^\nu, \e^{- \frac 7 {132} m  h_{I^L} (\lambda)\ \norm{y-a}}}\\   \notag 
& \le \e^{- \frac 7 {132}  m  h_{I^L} (\lambda)\norm{y-a}} \qtx{for all} y\in A_L(a).
\end{align}
 \end{proof}

\appendix

\section{Exponents}\label{apexp}

 Given  $ 0<\xi<\zeta<1$,  we consider  $ \beta,\tau \in (0,1)$ and $\gamma >1$ such that
\begin{gather}\label{ttauzeta0}
0<\xi< \zeta<\beta<\frac 1 \gamma <1<\gamma < \sqrt {\tfrac \zeta \xi} \mqtx{and}   \max\set{ \gamma \beta, \tfrac {(\gamma-1)\beta +1}{\gamma}  }      <  \tau <1;
\end{gather}
it follows that
\beq\label{ttauzeta}
0<\xi<\xi\gamma^2< \zeta < \beta <\frac \tau \gamma   <\frac 1 \gamma <\tau <1< \frac {1-\beta}{\tau-\beta} < \gamma <\frac \tau \beta.
\eeq

We set
\beq\label{ttauzeta2}
\tzeta= \frac {\zeta +\beta}2  \in (\zeta, \beta)  \qtx{and} {\ttau}= \frac {1 +\tau}2  \in (\tau,1),
\eeq
so
\beq\label{gamtzetabeta}
(\gamma-1)\tzeta  +1 <(\gamma-1)\beta  +1< \gamma \tau.
\eeq

We take $\kappa\in (0,1) $ and   $\kappa^\pr \in [0,1)$ such that
\beq\label{gamtzetabeta2}
 \kappa+\kappa^\pr < \tau - \gamma \beta.  
\eeq

We let 
\beq \label{defvrho}
\vrho= \min\set{\kappa, \tfrac{1- \tau}2, \gamma \tau- (\gamma-1)\tzeta  -1}, \qtx{note} 0<\kappa \le \vrho <1,
\eeq 
and choose
\beq\label{vsdef}
\vs \in (0,1-\vrho] , \qtx{so} \vrho < 1-\vs .
\eeq

  We select exponents satisfying \eq{ttauzeta0}- \eq{vsdef}  and fix these  exponents.

\section{Suitable covers of a box}\label{subsecsc}
To perform the multiscale analysis in an efficient way  we  use 
  suitable covers of a box as in  \cite[Section~3.4]{EK2}, an adaptation of   \cite[Definition~3.12]{GKber}.  We state the definition and properties for the reader's convenience.

\begin{definition}\label{defcov} Fix  $\vs \in (0,1)$.  Let $\La_L=\La_L(x_0)$, $x_0 \in \R^d$ be  a box in $\Z^d$, and let $\ell < L$.
A suitable $\ell$-cover of $\La_L$ is
the collection of  boxes  
\begin{align}\label{standardcover}
{\mathcal C}_{L,\ell}={\mathcal C}_{L,\ell} \left(x_0 \right)= \set{ {\Lambda}_{\ell}(a)}_{a \in  \Xi_{L,\ell}},
\end{align}
where
\beq  \label{bbG}
 \Xi_{L,\ell}= \Xi_{L,\ell}(x_0):= \set{ x_0+ {\rho}\ell^\vs  \Z^{d}}\cap \La_L^\R 
\mqtx{with}   {\rho}\in  \br{\tfrac {1} {2},1}   \cap \set{\tfrac {L-\ell}{2 \ell^{\vs} k}; \, k \in \N }.
\eeq
We call ${\mathcal C}_{L,\ell} $ the suitable $\ell$-cover of $\La_L$  if   ${\rho} ={\rho}_{L,\ell}: =\max \br{\tfrac {1} {2},1}   \cap \set{\tfrac {L-\ell}{2 \ell^{\vs} k}; \, k \in \N }.
$
\end{definition}

\begin{lemma}[{\cite[Lemma~3.13]{GKber}, \cite[Lemma~3.10]{EK2}}]\label{lemcover } Let $\ell \le \frac  L   2$. Then for  every  box   $\La_L=\La_L(x_0)$, $x_0 \in \R^d$, a suitable $\ell$-cover
  ${\mathcal C}_{L,\ell}={\mathcal C}_{L,\ell}\left(x_0\right) $  satisfies
  \begin{align}\label{nestingproperty} 
&\La_L=\bigcup_{a \in  \Xi_{L,\ell}} {\Lambda}_{\ell}(a);\\ \label{covproperty}
&\text{for all}\; \; b \in\La_L  \sqtx{there is} {\Lambda}_{\ell}^{(b)} \in {\mathcal C}_{L,\ell}  \sqtx{such that}  
  b\in \pa{ {\Lambda}_{\ell}^{(b)}}^{\La_L,\frac {\ell -\ell^\vs}2},\\ \notag
 & \qtx{i.e.,} \La_L=\bigcup_{a \in  \Xi_{L,\ell}} {\Lambda}_{\ell}^{\La_L, \frac {\ell -\ell^\vs}2}(a);
    \\ \label{number}
&   \# \Xi_{L,\ell}= \pa{ \tfrac{L-\ell} {{\rho} \ell^\vs}+1}^{d }\le   \pa{\tfrac{2L} {\ell^\vs}}^{d}.  
\end{align}
Moreover,  given $a \in x_0 + {\rho}\ell^\vs  \Z^{d}$ and $k \in \N$, it follows that
\beq \label{nesting}
{\Lambda}_{(2  k {\rho}\ell^\vs  + \ell)}(a)= \bigcup_{b \in  \{ x_0 + {\rho}\ell^\vs  \Z^{d}\}\cap {\Lambda}^{\R}_{(2k {\rho} \ell^\vs + \ell)}(a) } {\Lambda}_{\ell}(b),
\eeq
and  $ \{ \Lambda_{\ell}(b)\}_{b \in  \{ x_0 + {\rho}\ell^\vs  \Z^{d}\}\cap {\Lambda}^{\R}_{(2k {\rho} \ell^\vs + \ell)}(a) }$ is a suitable $\ell$-cover of the box $\Lambda_{(2k {\rho} \ell^\vs + \ell)}(a)$. 
\end{lemma}

Note that $\Lambda_{\ell}^{(b)}$  does not denote a box  centered at $b$, just some box in ${\mathcal C}_{L,\ell} \left(x_0 \right)$ satisfying \eq{covproperty}.     By  $\Lambda_{\ell}^{(b)}$ we will always mean such a box.  We will use
\beq
\dist\pa{b, {\partial^{\La_L}_{\mathrm{in}}\Lambda_{\ell}^{(b)}}}  \ge  \tfrac {\ell -\ell^\vs} 2 -1 \qtx{for all} b \in \La_L.
\eeq
Note also  that  ${\rho} \le 1$ yields \eq{covproperty}.   We specified ${\rho}={\rho}_{L,\ell}$ in  for \emph{the} suitable $\ell$-cover for convenience, so there is no ambiguity  in the definition of ${\mathcal C}_{L,\ell} \left(x_0 \right) $.

Suitable covers are convenient for the construction of buffered subsets (see Definition~\ref{defbuff}) in the multiscale analysis, where  we will assume   $\vs \in (0,1)$ is as in \eq{vsdef}.   We  will use the following observation:

\begin{remark} \label{remdisj}
 Let ${\mathcal C}_{L,\ell}$ be a suitable $\ell$-cover for the box $\La_L$, and set
$k_\ell= k_{L,\ell }=  \fl{\rho^{-1} \ell^{1-\vs}} +1$.
Then for all $a,b \in {\mathcal C}_{L,\ell}$ we have
\beq\label{disjbox}
\La_\ell^\R(a) \cap \La_\ell^\R(b) =\emptyset \quad \iff \quad \norm{a-b} \ge k_\ell \rho \ell^\vs .\eeq
\end{remark}


\begin{thebibliography}{AENSS}

\bibitem[A]{A}  Aizenman, M.: {Localization at weak disorder: some
elementary bounds}. Rev. Math. Phys. {\bf 6}, 1163-1182 (1994)

\bibitem[ASFH]{ASFH}  Aizenman, M.,  Schenker, J., Friedrich, R., 
Hundertmark, D.: Finite volume fractional-moment criteria for 
Anderson localization.
Commun. Math. Phys. \textbf{224}, 219-253 (2001) 

\bibitem[AENSS]{AENSS}  Aizenman, M., Elgart, A., Naboko, S.,   
 Schenker, J.,  Stolz, G.: Moment analysis for localization in random Schr\"odinger 
operators. Inv. Math. \textbf{163}, 343-413 (2006)

\bibitem[AM]{AM}  Aizenman, M.,  Molchanov, S.:  {Localization at large
disorder and extreme energies:  an elementary derivation}.  Commun. Math.
Phys. {\bf 157}, 245-278 (1993)

\bibitem[AW]{AW}  Aizenman, M.,  Warzel, S.: \emph{Random operators.
Disorder effects on quantum spectra and dynamics}.
Graduate Studies in Mathematics \textbf{168}. American Mathematical Society, Providence, RI,  2015.



\bibitem[An]{And}  Anderson, P.:  Absence of diffusion in certain random
lattices.  Phys. Rev. {\bf 109}, 1492-1505 (1958)

\bibitem[BK]{BK} Bourgain, J., Kenig, C.: On localization in the continuous Anderson-Bernoulli model in higher dimension.  Invent. Math. \textbf{161}, 389-426 (2005)

 
 \bibitem[CGK]{CGK2}  Combes,  J.M., Germinet, F.,  Klein, A.: Generalized eigenvalue-counting estimates for the Anderson model.  J. Stat. Phys. \textbf{135},   201-216 (2009). \doi{10.1007/s10955-009-9731-3}
 
 \bibitem[CH]{CH}  Combes,  J.M.,   Hislop, P.D.: {Localization for some
continuous, random Hamiltonians in d-dimension}. J. Funct. Anal. \textbf{124},
149-180 (1994)

\bibitem[CHK]{CHK} Combes,  J.M.,   Hislop, P.D.,  Klopp, F.:
 Optimal Wegner estimate and its application to the global continuity of the integrated
density of states for random Schr\"odinger operators.  Duke Math. J.~\textbf{140},
469-498 (2007)

\bibitem[DJLS1]{DRJLS0}  Del Rio, R.,  Jitomirskaya, S.,  Last, Y., Simon, B.:{ What is Localization?} Phys. Rev. Lett. 75, 117-119 (1995)

\bibitem[DJLS2]{DRJLS} Del Rio, R.,  Jitomirskaya, S.,  Last, Y., Simon, B.: {Operators with singular continuous spectrum IV: Hausdorff dimensions, rank one
perturbations and localization}. J. d'Analyse Math. {\bf 69}, 153-200 (1996)

\bibitem[DiEl]{DE}Dietlein ,A.,  Elgart, A.:  Level spacing for continuum random Schr\"odinger operators with applications,  \href{http://arxiv.org/abs/1712.03925}{arXiv:1712.03925}

 
 \bibitem[Dr]{Dr}  von Dreifus, H.: {\em On the effects of randomness in
ferromagnetic models and Schr\"odinger operators}.  Ph.D. thesis, New York
University (1987)
 
 \bibitem[DrK]{DK} von Dreifus, H.,  Klein, A.:  {A new proof of localization in
the Anderson tight binding model}.  Commun. Math. Phys. \textbf{124},
285-299  (1989). \\ \href{http://projecteuclid.org/euclid.cmp/1104179145}{http://projecteuclid.org/euclid.cmp/1104179145}


\bibitem[EK1]{EK}
Elgart, A., Klein, A.: An eigensystem approach to Anderson localization. J. Funct. Anal.  \textbf{271}, 3465-3512  (2016). \doi{10.1016/j.jfa.2016.09.008}

\bibitem[EK2]{EK2}
Elgart, A., Klein, A.: Eigensystem multiscale analysis for Anderson localization in energy intervals. {J. Spectr. Theory} \textbf{9}, 711-765 (2019).
   \doi{10.4171/JST/261}
   
   \bibitem[EK3]{EK3}
Elgart, A., Klein, A.:  In preparation

\bibitem[FK1]{FK1} Figotin, A.,  Klein, A.: Localization phenomenon in gaps of the spectrum of random lattice operators. J. Statist. Phys. \textbf{75}, 997-1021 (1994). \doi{10.1007/BF02186755}


\bibitem[FK2]{FK} Figotin, A.,  Klein, A.: {Localization of classical waves I:
Acoustic waves}.  Commun. Math. Phys. {\bf 180}, 439-482 (1996).  \href{http://projecteuclid.org/euclid.cmp/1104287356}{http://projecteuclid.org/euclid.cmp/1104287356}

\bibitem[FrS]{FS}  Fr\"ohlich, J.,  Spencer, T.: {Absence of diffusion with
Anderson tight binding model for large disorder or low energy}. Commun.
Math. Phys. {\bf 88}, 151-184 (1983)

\bibitem[FrMSS]{FMSS} Fr\"ohlich, J.:    Martinelli, F.,  Scoppola, E., 
Spencer, T.:  {Constructive proof of localization in the Anderson tight
binding model}. Commun. Math. Phys. {\bf 101}, 21-46 (1985)

\bibitem[GK1]{GKboot} Germinet, F.,  Klein, A.: {Bootstrap multiscale analysis and
localization in random media}.  Commun. Math. Phys. \textbf{222},
415-448 (2001). \doi{10.1007/s002200100518}

\bibitem[GK2]{GKsudec} Germinet, F., Klein, A.: New characterizations of the region of complete localization for random Schr\"odinger operators.  J. Stat. Phys. \textbf{122}, 73-94 (2006). \doi{10.1007/s10955-005-8068-9}


\bibitem[GK3]{GKber} Germinet, F.,  Klein, A.:  A comprehensive proof of localization for continuous Anderson models with singular random potentials.  J. Eur. Math. Soc. \textbf{15}, 53-143 (2013). \doi{10.4171/JEMS/356}



%

\bibitem[K]{Ki} Kirsch, W.: {An invitation to random Schr\"odinger operators. In Random Schr\"odinger-Operators. Panoramas et Syntheses \textbf{25}}. 
Societe Mathematique de France, Paris, 1-119 (2008)
%

\bibitem[Kl1]{Kle} Klein, A.:   Multiscale analysis and localization of random operators.
 In \emph{Random Schr\" odinger Operators}.  Panoramas et Synth\`{e}ses \textbf{25},   121-159,  
 Soci\'{e}t\'{e}  Math\'{e}matique de France, Paris 2008
 
 \bibitem[Kl2]{Kl2}  Klein, A.: Unique continuation principle for spectral projections of Schr\" odinger operators and optimal Wegner estimates for non-ergodic random Schr\" odinger operators.  Comm. Math Phys. \textbf{323},  1229-1246 (2013).  \doi{10.1007/s00220-013-1795-x}


\bibitem[KlM]{KlM} Klein. A., Molchanov, S.: Simplicity of eigenvalues in the Anderson
model.  J. Stat. Phys.~{\bf 122}, 95-99 (2006). \doi{10.1007/s10955-005-8009-7}

\bibitem[KlT]{KlT} Klein. A., Tsang, C.S.S.: Eigensystem bootstrap multiscale analysis for the Anderson model. J. Spectr. Theory \textbf{8}, 1149-1197 (2018). \doi{10.4171/JST/224}



\bibitem[M]{M} Minami, N.: Local fluctuation of the spectrum of a multidimensional
Anderson tight binding model. Commun. Math. Phys.~{\bf 177}, 709-725 (1996)



\bibitem[S]{Sp}  Spencer,  T.:  Localization for random and
quasiperiodic potentials.    J. Stat. Phys. {\bf 51}, 1009-1019 (1988)


\bibitem[W]{Weg} Wegner, F.: Bounds on the density of states in disordered systems, {Z.
Phys. B}{\bf 44} {9-15} (1981)


\end{thebibliography}
\end{document}